\documentclass[11pt,english]{article}

\usepackage{mathtools}
\usepackage{amsmath}
\usepackage{amssymb}
\usepackage{amsthm}
\usepackage{algorithm,algpseudocode}
\usepackage{xcolor}
\usepackage{fullpage}
\usepackage{babel}
\usepackage{amsthm}
\usepackage{verbatim}
\usepackage{thmtools}
\usepackage{thm-restate}
\PassOptionsToPackage{normalem}{ulem}
\usepackage{ulem}

\newcommand{\wt}{\mathrm{wt}}

\theoremstyle{plain}

\theoremstyle{plain}
\newtheorem{theorem}{Theorem}[section]
\theoremstyle{plain}

\theoremstyle{plain}
\newtheorem{theorem-seq}{Theorem}
 \theoremstyle{definition}
 \newtheorem{definition}[theorem]{Definition}
 \theoremstyle{definition}
 
 \theoremstyle{definition}
 \newtheorem{remark}[theorem]{Remark}
 \theoremstyle{definition}
 \newtheorem{rem}[theorem]{Remark}
 \theoremstyle{plain}
 \newtheorem{lemma}[theorem]{Lemma}  
 \theoremstyle{plain}
 \newtheorem{lem}[theorem]{Lemma}  
 \theoremstyle{plain}
 \newtheorem{fact}[theorem]{Fact}
 \theoremstyle{plain}
 \newtheorem{corollary}[theorem]{Corollary}  
 \theoremstyle{plain}
 \theoremstyle{definition}
 \newtheorem*{acknowledgement*}{Acknowledgement}
 \theoremstyle{plain}
 \theoremstyle{plain}
\theoremstyle{plain}

\def\notes{1}
 \newcommand{\nnote}[1]{\ifnum\notes=1{{\sf\color{blue} [Noga comment: #1]}}\fi}

\usepackage{xcolor,amsfonts}

\newcommand{\cC}{C} 

\newcommand{\F}{{\mathbb F}}

\newcommand{\inset}[1]{\left\{#1\right\}}

\newcommand{\inparen}[1]{\left(#1\right)}

\newcommand{\suchthat}{\,:\,}

\newcommand{\polylog}{\mathrm{polylog}}
\newcommand{\poly}{\mathrm{poly}}

\newcommand{\dist}{\ensuremath{\operatorname{dist}}}

\newcommand{\eps}{\varepsilon}
\renewcommand{\epsilon}{\varepsilon}

\newcommand{\subp}[2]{\ensuremath{#1}_{(#2)}}

\newcommand{\Tpre}{\ensuremath{T_{\text{pre}}}}

\newcommand{\Ttpre}{\ensuremath{\widetilde{T}_{\text{pre}}}}
\newcommand{\defined}{\ensuremath{:=}}
\algnewcommand{\LineComment}[1]{\State \(\triangleright\) #1}

    \global\long\def\inn{in}
  \global\long\def\out{out}

\title{Local List Recovery of High-rate Tensor Codes \& Applications}

\author{Brett Hemenway%
\thanks{\texttt{fbrett@cis.upenn.edu}. Department of Computer Science, University of Pennsylvania}
\and Noga Ron-Zewi%
\thanks{\texttt{nogazewi@cs.bgu.ac.il}. Department of Computer Science, Ben-Gurion University }
\and Mary Wootters%
\thanks{\texttt{marykw@stanford.edu}. Departments of Computer Science and Electrical Engineering, Stanford University}
}

\begin{document}

\maketitle

\begin{abstract}
In this work, we give the first construction of {\em high-rate}  locally list-recoverable codes.  List-recovery has been an extremely useful building block in coding theory, and our motivation is to use these codes as such a building block.
 In particular, our construction gives the first {\em capacity-achieving} locally list-decodable codes (over constant-sized alphabet); the first {\em capacity achieving} globally list-decodable codes with nearly linear time list decoding algorithm (once more, over constant-sized alphabet); and a randomized construction of binary codes on the Gilbert-Varshamov bound that can be uniquely decoded in near-linear-time, with higher rate than was previously known.

Our techniques are actually quite simple, and are inspired by an approach of Gopalan, Guruswami, and Raghavendra (Siam Journal on Computing, 2011) for list-decoding tensor codes.  We show that tensor powers of (globally) list-recoverable codes are `approximately' locally list-recoverable, and that the `approximately' modifier may be removed by pre-encoding the message with a suitable locally decodable code.  Instantiating this with known constructions
of high-rate globally list-recoverable codes
and high-rate locally decodable codes finishes the construction.
\end{abstract}

\section{Introduction}\label{sec:intro}

\em List-recovery \em refers to the problem of decoding error correcting codes from ``soft" information.
More precisely, given a \em code \em $\cC: \Sigma^k \to \Sigma^n$, which maps length-$k$ \em messages \em to length-$n$ \em codewords, \em
 an $(\alpha, \ell, L)$-list-recovery algorithm for $C$ is provided with a sequence
of lists $S_1,\ldots, S_n \subset \Sigma$ of size at most $\ell$ each, and is tasked with efficiently returning all messages $x \in \Sigma^k$ so that $C(x)_i \notin S_i$ for at most $\alpha$ fraction of the coordinates $i$; the guarantee is that there are no more than $L$ such messages.  The goal is to design codes $\cC$ which simultaneously admit such algorithms, and which also have other desirable properties, like high \em rate \em (that is, the ratio $k/n$, which captures how much information can be sent using the code) or small \em alphabet size \em $|\Sigma|$.
List-recovery is a generalization of \em list-decoding, \em which is the situation when the lists $S_i$ have size one: we refer to $(\alpha, 1, L)$-list-recovery as $(\alpha, L)$-list-decoding.

	List recoverable codes were first studied in the context of list-decoding and soft-decoding.  The celebrated Guruswami-Sudan list-decoding algorithm~\cite{GS99} is in fact a list-recovery algorithm, as are several more recent list-decoding algorithms~\cite{GR08_folded_RS,GW11,Kop15,GX13}.
	Initially, list recoverable codes were used as stepping stones towards constructions of list decodable and uniquely decodable codes \cite{GI01,GI02,GI03,GI04}.
	Since then, list recoverable codes have found additional applications in the areas of compressed sensing, combinatorial group testing, and hashing \cite{INR10,NPR12,GNPRS13,HIOS}.

\em Locality \em is another frequent desideratum in coding theory.  Loosely, an algorithm is ``local" if information about a single coordinate $x_i$ of a message $x$ of $\cC$ can be determined locally from only a few coordinates of a corrupted version of $C(x)$.  Locality, and in particular local list-decoding, has been implicit in theoretical computer science for decades:
for example, local list-decoding algorithms are at the heart of algorithms in cryptography~\cite{GL89}, learning theory~\cite{KM93}, and hardness amplification and derandomization~\cite{STV01}.

The actual definition of a locally list-recoverable (or locally list-decodable) code requires some subtlety: we want to return a list of answers, and we want the algorithm to be local (having to do with a single message coordinate), but returning a list of possible symbols in $\Sigma$ for a single message coordinate is pretty useless if that was our input to begin with (at least if the code is systematic in which case the message coordinates are a subset of the codeword coordinates).  Instead, we require that a local list-recovery algorithm returns a list $A_1,\ldots, A_L$ of randomized local algorithms.  Each of these algorithms takes an index $i \in [k]$ as input, and has oracle access to the lists $S_1,\ldots, S_n$.  The algorithm then makes at most $Q$ queries to this oracle (that is, it sees at most $Q$ different lists $S_i$), and must return a guess for $x_i$, where $x$ is a message whose encoding $C(x)$ agrees with many of the lists.  The guarantee is that for all such $x$---that is, for all $x$ whose encoding $C(x)$ agrees with many of the lists---there exists (with high probability) some $A_j$ so that for all $i$, 
 $A_j(i) = x_i$ with probability at least $2/3$.  The parameter $Q$ is called the \em query complexity \em of the local list-recovery algorithm.

One reason to study local list-recoverability is that list-recovery is a very useful building block throughout coding theory.  In particular, the problem of constructing \textbf{high rate locally list-recoverable codes} (of rate arbitrarily close to $1$, and in particular non-decreasing as a function of $\ell$) has been on the radar for a while, because such codes would have implications in local list-decoding, global list-decoding, and classical unique decoding.

In this work, we give the first constructions of high-rate locally list-recoverable codes.  As promised, these lead to several applications throughout coding theory.  Moreover, our construction is actually quite simple.  Our approach is inspired by the list-decoding algorithm of~\cite{GGR} for tensor codes, and our main observation is that this algorithm---with a few tweaks---can be made local.

\subsection{Results}
We highlight our main results below---we will elaborate more on these results and their context within related literature next in Section~\ref{sec:related}.
\begin{description}
\item[High-rate local list-recovery.]
Our main technical contribution is the first constructions of high-rate locally list-recoverable codes: Theorems~\ref{thm:nonexplicit-high-rate-local-list-rec} and \ref{thm:explicit-high-rate-local-list-rec} give the formal statements.  Theorem~\ref{thm:nonexplicit-high-rate-local-list-rec} can guarantee high-rate  list recovery with query complexity $n^{1/t}$ (for constant $t$, say $n^{0.001}$), constant alphabet size and \em constant \em output list size, although without an explicit construction or an efficient list recovery algorithm.  Theorem~\ref{thm:explicit-high-rate-local-list-rec} on the other hand gives an explicit and efficient version, at the cost of a slightly super-constant output list size (which depends on $\log^*n$).

For those familiar with the area, it may be somewhat surprising that this was not known before: indeed, as discussed below in Section~\ref{sec:related}, we know of locally list recoverable codes (of low rate), and we also know of high-rate (globally) list-recoverable codes. One might think that our result is lurking implicitly in those earlier works.
However, it turns out that it is not so simple: as discussed below, existing techniques for locally or globally list-recoverable codes do not seem to work for this problem.  Indeed, some  of those prior works~\cite{HW15, KMRS15-STOC, GKORS17} (which involve the current authors) began with the goal of obtaining high-rate locally list-recoverable codes and ended up somewhere else.

This raises the question: why might one seek high-rate locally list-recoverable error correcting codes in the first place?  The motivation is deeper than a desire to add adjectives in front of ``error correcting codes."  As we will see below, via a number of reductions that already exist in the literature, such codes directly lead to improvements for several other fundamental problems in coding theory, including fast or local algorithms for list and unique decoding.

\item[Capacity-achieving locally list-decodable codes.]
The first such reduction 
is an application of an expander-based technique of Alon, Edmunds, and Luby~\cite{AEL95}, which allows us to turn the high-rate locally list-recoverable codes into \em capacity achieving \em locally list-decodable (or more generally, locally list recoverable) codes.  Our main results are stated as Theorems~\ref{thm:nonexplicit-cap-local-list-rec} and \ref{thm:explicit-cap-local-list-rec}.

As before, Theorem~\ref{thm:nonexplicit-cap-local-list-rec} gives capacity achieving locally list decodable codes
with query complexity $n^{0.001}$ (say), constant alphabet size and constant  output list size, although without an explicit construction or efficient list decoding algorithms.
Theorem~\ref{thm:explicit-cap-local-list-rec} on the other hand gives explicit and efficiently list decodable codes, and a trade-off between query complexity, alphabet size, and output list size.  Specifically, these codes obtain query complexity $Q = n^{1/t}$ with an output list size and an alphabet size that grow doubly exponentially with $t$ (and output list size depends additionally on $
\log^*n$).
In particular, if we choose $t$ to be constant, we obtain query complexity $n^{1/t}$, with constant alphabet size and nearly-constant output list size.  We may also choose to take $t$ to be very slowly growing, and this yields query complexity $n^{o(1)}$, with output list and alphabet size $n^{o(1)}$ as well. Prior to this work, no construction of capacity achieving locally list decodable codes with query-complexity $o(n)$ was known.

\item[Near-linear time capacity-achieving list-decodable codes.]
Given an efficiently list decodable capacity achieving locally list-decodable code (as given by Theorem~\ref{thm:explicit-cap-local-list-rec} mentioned above), it is straightforward to construct fast algorithms for global list-decoding the same code.  Indeed, we just repeat the local decoding algorithm (which can be done in time $n^{O(1/t)}$) a few times, for all $n$ coordinates, and take majority vote at each coordinate.  Thus, our previous result implies
explicit, capacity-achieving, list-decodable codes (or more generally, list recoverable codes) that can be (globally) list-decoded (or list-recovered) in time $n^{1 + O(1/t)}$.

This result is reported in Theorem~\ref{thm:near-linear-cap-list-rec}.  As with the previous point, this result actually allows for a trade-off: we obtain either decoding time $N^{1.001}$ (say) with constant alphabet size and near-constant output list size, or
 decoding time $n^{1 + o(1)}$ at the cost of increasing the alphabet and output list size to $n^{o(1)}$. Previous capacity achieving list-decoding algorithms required at least quadratic time for recovery.

\item[Near-linear time unique decoding up to the Gilbert Varshamov bound.]
Via a technique of Thommesen~\cite{Thomm83} and Guruswami and Indyk~\cite{GI04}, our near-linear time capacity-achieving list-recoverable codes give a randomized construction of low-rate (up to $0.02$) binary codes approaching the \em Gilbert-Varshamov \em (GV) bound, which admit near-linear time ($n^{1 + o(1)}$) algorithms for unique decoding up to half their distance.  The formal statement is in Theorem~\ref{thm:near-linear-GV}.  Previous constructions which could achieve this either required at least quadratic decoding time, or else did not work for rates larger than $10^{-4}$.
\end{description}

Our approach (discussed more below) is modular; given as an ingredient any (globally) high-rate list-recoverable code
(with a polynomial time recovery algorithm), it yields high-rate (efficiently) locally list-recoverable code.  To achieve the results advertised above, we instantiate this with either a random (non-efficient) linear code or with
 the (efficient) Algebraic Geometry (AG) subcodes of~\cite{GK14}.  Any improvements in these ingredient codes (for example, in the output list size of AG codes, which is near-constant but not quite) would translate immediately into improvements in our constructions.

\paragraph{Organization.}  In Section~\ref{sec:related} below, we discuss related work and put our results in context, followed by an overview of our techniques in Section~\ref{sec:techniques}.  We set up notation and definitions in Section~\ref{sec:prelims}.  Our main technical contribution is a construction of high-rate locally list-recoverable codes, and this is handled in Section~\ref{sec:locallistrecovery}.  In Section~\ref{sec:caplldc}, we discuss the implication to capacity-achieving local list decoding.  Finally, Section~\ref{sec:lintime} discusses our applications to near-linear-time decoding algorithms: Section~\ref{subsec:near-linear-cap-list-dec} presents our results for global capacity-achieving list-decoding, and Section~\ref{subsec:near-linear-GV} presents our results for unique decoding up to the Gilbert-Varshamov bound.

\section{Related work}\label{sec:related}
As mentioned above,  list decoding and recovery, local decoding, and local list decoding and recovery, have a long and rich history in theoretical computer science.  We mention here  the results that are most directly related to ours mentioned above.

\paragraph{High-rate local list recovery.}  Our main technical contribution is the construction of high-rate locally list-recoverable codes.
There are two lines of work that are most related to this: the first is on local list recovery, and the second on high-rate (globally) list-recoverable codes.

Local list-decoding (which is a special case of local list recovery) first arose outside of coding theory, motivated by applications in complexity theory.  For example, the Goldreich-Levin theorem in cryptography and the Kushilevitz-Mansour algorithm in learning theory are a local list-decoding algorithm for Hadamard codes.  Later, Sudan, Trevisan and Vadhan~\cite{STV01}, motivated by applications in pseudorandomness, gave an algorithm for locally list-decoding Reed-Muller codes.
Neither Hadamard codes nor the Reed-Muller codes of~\cite{STV01} are high-rate.   However, similar ideas can be used to locally list-decode \em lifted codes \em \cite{GK16}, and \em multiplicity codes \em \cite{Kop15}, which can be seen as high-rate variants of Reed-Muller codes. These algorithms work up to the so-called \em Johnson bound. \em

Briefly, the Johnson bound says that a code of distance $\delta$ is $(\alpha, L)$-list-decodable, for reasonable $L$, when $\alpha \leq 1 - \sqrt{1 - \delta}$.  This allows for high rate list decodable codes when $\delta$ is small, but there exist codes which are more list-decodable: the \em list-decoding capacity theorem \em implies that there are codes of distance $\delta$ which are $(\alpha, L)$-list-decodable for $\alpha$ approaching the distance $\delta$.  The ``capacity-achieving" list-decodable codes that we have been referring to are those which meet this latter result, which turns out to be optimal.

Like many list-decoding algorithms, the algorithms of \cite{STV01, Kop15, GK16} can be used for list-recovery as well (indeed, this type of approach was recently used in \cite{GKORS17} to obtain a local list-recovery algorithm for Reed-Muller codes.) However, as mentioned above they only work up to the Johnson bound for list-decoding, and this holds for list-recovery as well.  However, for list-recovery, the difference between the Johnson bound and capacity is much more stark.  Quantitatively,  for $(\alpha, \ell, L)$-list-recovery, the Johnson bound requires $\alpha \leq 1 - \sqrt{\ell(1 - \delta)}$, which is meaningless unless $\delta$ is very large; this requires the rate of the code to be small, less than $1/\ell$.   In particular, these approaches do not give high-rate codes for list-recovery, and the Johnson bound appears to be a fundamental bottleneck.

The second line of work relevant to high-rate local list-recovery is that on high-rate \em global \em list-recovery.   Here, there are two main approaches.  The first is a line of work on capacity achieving list-decodable codes (also discussed more below).  In many cases, the capacity achieving list-decoding algorithms for these codes are also high-rate list-recovery algorithms~\cite{GR08_folded_RS,GW11,Kop15,GX13}.  These algorithms are very global: they are all based on finding some interpolating polynomial, and finding this polynomial requires querying almost all of the coordinates.  Thus, it is not at all obvious how to tweak these sorts of algorithms to achieve locally list-recoverable codes.  The other line of work on high-rate global list-recovery is that of~\cite{HW15}, which studies high-rate list-recoverable expander codes.  While that algorithm is not explicitly local, it's not as clearly global as those previously mentioned (indeed, expander codes are known to have some locality properties~\cite{HOW13}).  However, that work could only handle list-recovery with no errors---that is, it returns codewords that agree with {\em all} of the lists $S_i$, rather than a large fraction of them---and adapting it to handle errors seems like a challenging task.

Thus, even with a great deal of work on locally list recoverable codes, and on high-rate globally list-recoverable codes, it was somehow not clear how to follow those lines of work to obtain high-rate locally list-recoverable codes.   Our work, which does give high-rate locally list-recoverable codes, follows a different approach, based on the techniques of~\cite{GGR} for list-decoding tensor codes.  In fact, given their ideas and a few other ingredients, our solution is actually quite simple!  We discuss our approach in more detail in Section~\ref{sec:techniques}.

\paragraph{Capacity achieving locally list decodable codes.}
As mentioned above, one reason to seek high-rate codes is because of a transformation of Alon, Edmunds, and Luby~\cite{AEL95}, recently highlighted in~\cite{KMRS15-STOC}, which can, morally speaking, turn any high-rate code with a given property into a capacity achieving code with the same property.\footnote{We note however that this transformation does not apply to the property of list decoding, but just list recovery, and therefore we cannot use existing constructions of high-rate locally list decodable codes \cite{Kop15, GK16} as a starting point for this transformation.}
This allows us to obtain \emph{capacity achieving} locally list-decodable (or more generally, locally list recoverable) codes.
This trick has been used frequently over the years~\cite{GI01,GI02,GI03,GI04,HW15,KMRS15-STOC,GKORS17}, and in particular \cite{GKORS17} used it for local list recovery.  We borrow this result from them, and this immediately gives our  capacity achieving locally list-decodable  codes.  Once we have these, they straightforwardly extend to near-linear time capacity-achieving (globally) list-decodable (or more generally, locally list recoverable) codes, simply by repeatedly running the local algorithm on each coordinate.

\paragraph{Capacity-achieving list-decodable codes.}
We defined list-decodability above as a special case of list-recovery, but it is in fact much older.  List-decodability has been studied since the work of Elias and Wozencraft \cite{E57,W58} in the late 1950s, and the combinatorial limits are well understood.
The \em list-decoding capacity theorem, \em mentioned earlier, states that there exist codes of rate approaching $1 - H_q(\alpha)$ which are $(\alpha, L)$-list-decodable for small list size $L$, where $H_q(\alpha)$ is the $q$-ary entropy function (when $q$ is large we have $1-H_q(\alpha)\approx 1-\alpha$).  Moreover, any code of rate larger than that must have exponentially large list size.

The existence direction of the list-decoding capacity theorem follows from a random coding argument, and it wasn't until the Folded Reed-Solomon Codes of Guruswami and Rudra~\cite{GR08_folded_RS} that we had explicit constructions of codes which achieved list-decoding capacity.  Since then, there have been many more constructions~\cite{Gur10, GW11, Kop15, DL12, GX12, GX13, GK14}, aimed at reducing the alphabet size, reducing the list size, and improving the speed of the recovery algorithm.  We show the state-of-the-art in Table~\ref{table:listdec} below, along with our results (Theorem~\ref{thm:near-linear-cap-list-rec}).

\def\arraystretch{1.5}
\begin{table}[h]
\footnotesize
\begin{tabular}{|p{3cm}|p{2cm}|c|c|c|c|}
\hline
Code & Reference & Construction& Alphabet size & List size & Decoding time \\
\hline
Folded RS codes, derivative codes & \cite{GR08_folded_RS,GW11,Kop15} & Explicit & $\poly(n)$ & $\poly(n)$ & $n^{O(1/\eps)}$ \\
\hline
Folded RS subcodes & \cite{DL12} & Explicit & $\poly(n)$ & $O(1)$ & $n^{2} $ \\
\hline
(Folded) AG subcodes & \cite{GX12,GX13} & Monte Carlo & $O(1)$ & $O(1)$ & $n^{c} $ \\
\hline
AG subcodes&\cite{GK14} & Explicit & $O(1)$ & $\exp(\exp( (\log^*n)^2))$ & $n^{c} $ \\
\hline
Tensor codes&Theorem~\ref{thm:near-linear-cap-list-rec} & Explicit & $O(1)$ & $\exp(\exp(\exp(\log^*n)))$ & $n^{1.001}$ \\
\hline
\end{tabular}
\caption{Constructions of list-decodable codes that enable $(\alpha,L)$ list decoding up to rate $\rho = 1 - H_q(\alpha) - \eps$, for constant $\eps$.  We have suppressed the dependence on $\eps$, except where it appears in the exponent on n in the decoding time.
Above, $c$ is an unspecified constant.  In the analysis of these works, it is required to take $c \geq 3$.  It may be that these approaches could be adapted (with faster linear-algebraic methods) to use a smaller constant $c$, but it is not apparent; in particular we cannot see how to take $c < 2$.}
\label{table:listdec}
\end{table}

\paragraph{Unique decoding up to the Gilbert-Varshamov bound.}
The \em Gilbert-Varshamov \em (GV) bound \cite{G52,V57}  is a classical achievability result in coding theory.  It states that there exist binary codes of relative distance $\delta \in (0,1)$ and rate $\rho$ approaching $1- H_2(\delta),$ where $H_2$ is the binary entropy function.  The proof is probabilistic: for example, it is not hard to see that a random linear code will do the trick.  However, finding \em explicit \em constructions of codes approaching the GV bound remains one of the most famous open problems in coding theory.  While we cannot find explicit constructions, we may hope for randomized constructions with efficient algorithms, and indeed this was achieved in the low-rate regime through a few beautiful ideas by Thommesen~\cite{Thomm83} and follow-up work by Guruswami and Indyk \cite{GI04}.

Thommesen gave an efficient randomized construction of \em concatenated codes \em approaching the GV bound.  Starting with a Reed-Solomon code over large alphabet, the construction is to concatenate each symbol with an independent random linear code.  Later, \cite{GI04} showed that these codes could in fact be efficiently decoded up to half their distance, in polynomial time, up to rates about $10^{-4}$.  Their idea was to use the list recovery properties of Reed-Solomon codes.  The algorithm is then to list decode the small inner codes by brute force, and to run the efficient list-recovery algorithm for Reed-Solomon on the output lists of the inner codes: the combinatorial result of Thommesen ensures that the output list will contain a message that corresponds to the transmitted codeword.

In their work, \cite{GI04} used the Guruswami-Sudan list recovery algorithm~\cite{GS99}.  After decades of work~\cite{A02, BB10,CH11,BHNW13,CJNSV15}, this algorithm can now be implemented to run in near-linear time, and so we already can achieve near-linear time unique decoding near the GV bound, up to rates about $10^{-4}$.
The reason for the bound on the rate is that the Guruswami-Sudan algorithm only works up to the aforementioned Johnson bound, which means it cannot tolerate as much error as capacity-achieving list-recoverable codes.  It was noted by Rudra \cite{R07a} that replacing the Reed-Solomon codes with a capacity achieving list recoverable code (such as folded Reed-Solomon codes) can improve this rate limit up to about $0.02$.   However, those capacity achieiving list recovery algorithms were slower (as in Table~\ref{table:listdec}), and this increases the running time back to at best quadratic.

The recent work \cite{GKORS17} also applied these techniques to give \em locally decodable codes \em approaching the Gilbert-Varshamov bound.  These have query complexity  $n^\beta$, and so in particular can be easily adapted to give a global decoding algorithm with running time $O(n^{1 + \beta})$.  However, the rate up to which the construction works approaches zero exponentially quickly in $1/\beta$.

Using exactly the same approach as these previous works, we may plug in our capacity achieving near-linear-time list-recooverable codes to obtain binary codes approaching the GV bound, which are uniquely decodable up to half their distance in time $n^{1 + o(1)}$, and which work with rate matching Rudra's, $\rho = 0.02$.

\begin{remark}
It is natural to ask whether our result can, like \cite{GKORS17}, give locally decodable codes on the GV bound of higher rate.   In fact, we do not know how to do this.  The main catch is that our locality guarantees are for local decoding rather than local correction.  That is, we can only recover the message symbols and not the codeword symbols, and consequently we do not know how to choose the message from the output list that corresponds to the unique closest codeword.  It is an interesting open question whether one can use our techniques to extend the results of \cite{GKORS17} to higher rates.
\end{remark}

\paragraph{List-decodability and local properties of tensor codes.}
As we elaborate on more below in Section~\ref{sec:techniques}, our codes are constructed by taking tensor products of existing constructions of globally list-recoverable codes.
Our approach is inspired by that of \cite{GGR}, who study the list-decodability of tensor codes, although they do not address locality.  It should be noted that the local \em testing \em properties of tensor codes have been extensively studied~\cite{BS-tensor,V05_tensor_product,CR05,DSW06,GM12,BV09,BV15,V11_tensor_robustness,M09_LTCs,V13}.
To the best of our knowledge, ours is the first work to study the local (list) \em decodability \em of tensor codes, rather than local testability.

\section{Overview of techniques}\label{sec:techniques}
Our main technical contribution is the construction of high-rate locally list-recoverable codes.  While these are  powerful objects, and result in new sub-linear and near-linear time algorithms for fundamental coding theoretic tasks, our techniques are actually quite simple (at least if we take certain previous works as a black box).  We outline our approach below.

Our main ingredient is \em tensor codes, \em and the analysis given by Gopalan, Guruswami, and Raghevendra in~\cite{GGR}.
Given a linear code $C: \F^k \to \F^n$, consider the \em tensor code \em $C \otimes C: \F^{k \times k} \to \F^{n \times n}$; we will define the tensor product formally in Definition~\ref{defn:tensor}, but for now, we will treat the codewords of  $C \otimes C$ as $n\times n$ matrices with the constraints that the rows and columns are all codewords of the original code $C$.

In \cite{GGR}, it is shown that the tensor code $C \otimes C$ is roughly as list-decodable as $C$ is.
That work was primarily focused on combinatorial results, but their techniques are algorithmic, and it is these algorithmic insights that we leverage here.
The algorithm is very simple: briefly, we imagine fixing some small combinatorial rectangle $S \times T \subseteq [n] \times [n]$ of ``advice."   Think of this advice as choosing the symbols of the codeword indexed by those positions.  By alternately list decoding rows and columns, it can be shown that this advice uniquely determines a codeword $c$ of $C \otimes C$.  Finally, iterating over all possible pieces of advice yields the final list.

Inspired by their approach,  our Main Technical Lemma~\ref{approx-local-list-rec} says that if $C$ is list-recoverable, then not only $C \otimes C$ is also list-recoverable, but in fact it is (approximately) locally list-recoverable.
To understand the intuition, let us describe the algorithm just for $C \otimes C$, although our actual codes will require a higher tensor power $C^{\otimes t}$.
Suppose that $C$ is list-recoverable with output list size $L$.
First, imagine fixing some advice $J:=(j_1,\ldots, j_m)\in [L]^m$ for some (small integer) parameter $m$.   This advice will determine an algorithm $\widetilde{A}_J$ which attempts to locally decode some message that corresponds to a close-by codeword $c$ of $C \otimes C$, and the list we finally return will be the list of all algorithms $\widetilde{A}_J$ obtained by iterating over all possible advice.

Now, we describe the randomized algorithm $\widetilde{A}_J$, on input $(i,i') \in [n] \times [n]$.\footnote{The algorithm $\widetilde{A}_J$ we describe decodes codeword bits instead of messages bits, but since the codes we use are systematic this algorithm can also decode message bits.}  Recall, $\widetilde{A}_J$ is allowed to query the input lists at every coordinate, and must produce a guess for the codeword value indexed by $(i,i')$.    First, $\widetilde{A}_J$ chooses $m$ random rows of $[n] \times [n]$.  These each correspond to codewords in $C$, and $\widetilde{A}_J$ runs $C$'s list-recovery algorithm on them to obtain lists $\mathcal{L}_1, \ldots,\mathcal{L}_m$ of size at most $L$ each.  Notice that this requires querying $ mn$ coordinates, which is roughly the square root of the length of the code (which is $n^2$).  Then, $\widetilde{A}_J$ will use the advice $j_1,\ldots,j_m$ to choose codewords from each of these lists, and we remember the $i'$'th symbol of each of these codewords.  Finally, $\widetilde{A}_J$ again runs $C$'s list-recovery algorithm on the $i'$'th column, to obtain another list $\mathcal{L}$.  Notice that our advice now has the same form as it does in \cite{GGR}: we have chosen a few symbols of a codeword of $C$.  Now $\widetilde{A}_J$ chooses the codeword in $\mathcal{L}$ that agrees the most with this advice.  The $i'$'th symbol of this codeword is $\widetilde{A}_J$'s guess for the $(i,i')$ symbol of the tensor codeword.

The above idea gives a code of length $n$ which is locally list-recoverable with query complexity on the order of $\sqrt{n}$.  This algorithm for $C \otimes C$ extends straightforwardly to $C^{\otimes t}$, with query complexity $n^{1/t}$.  The trade-off is that the output list-size also grows with $t$.  Thus, as we continue to take tensor powers, the locality improves, while the output list-size degrades; this allows for the trade-off between locality and output list-size mentioned in the introduction.

One issue with this approach is that this algorithm may in fact fail on a constant fraction of coordinates $(i,i')$ (e.g., when a whole column is corrupted).  To get around this, we first encode our message with a high-rate \em locally decodable code, \em before encoding it with the tensor code.  For this, we use the codes of~\cite{KMRS15-STOC}, which have rate that is arbitrarily close to $1$, and which are locally decodable with $\exp(\sqrt{\log n})$ queries.  This way, instead of directly querying the tensor code (which may give the wrong answer a constant fraction of the time), we instead use the outer locally decodable code to query the tensor code: this still does not use too many queries, but now it is robust to a few errors.

The final question is what to use as an inner code.  Because we are after high-rate codes we require $C$ to be high-rate (globally) list recoverable. Moreover, since the tensor operation inflates the output list size by quite a lot, we require $C$ to have small (constant or very slowly growing) output list size. Finally, we need $C$ to be linear to get a handle on the rate of the tensor product. One possible choice is random linear codes, and these give a non-explicit and non-efficient construction with constant output list size. Another possible choice is
 the Algebraic Geometry subcodes of~\cite{GX13, GK14} which give explicit and efficient construction but with slowly growing output list size.  However, we cannot quite use these latter codes as a black box, for two reasons.  First, the analysis in~\cite{GX13} only establishes list-\emph{decodability}, rather than list-recoverability.  Fortunately, list-recoverability follows from exactly the same argument as list-decodability.  Second, these codes are linear over a subfield, but are not themselves linear, while our arguments require linearity over the whole alphabet.  
 Fortunately, we can achieve the appropriate linearity by concatenating the AG subcode with a small list-recoverable linear code, which exists by a probabilistic argument.  We handle these modifications to the approach of~\cite{GX13,GK14} in Appendix~\ref{app:ag}.

To summarize, our high-rate locally list-recoverable code is given by these ingredients: to encode a message $x$, we first encode it with the \cite{KMRS15-STOC} locally decodable code.  Then we encode this with a $t$-fold tensor product of a random linear code or a modified AG subcode and we are done.
We go through the details of the argument sketched above in Section \ref{sec:locallistrecovery}; but first, we introduce some notation and formally define the notions that we will require.

\section{Definitions and preliminaries}\label{sec:prelims}

For a prime power $q$ we denote by $\F_{q}$ the finite field of $q$ elements.
For any finite alphabet~$\Sigma$ and for any pair of strings $x,y\in\Sigma^{n}$, the \textsf{relative
distance} between $x$ and $y$ is the fraction of coordinates $i \in[n]$ on which $x$ and $y$
differ, and is denoted by $\dist(x,y):= \left|\left\{ i\in\left[n\right]:x_{i}\ne y_{i}\right\} \right|/n$. For a positive integer $\ell$ we denote
by ${\Sigma \choose \ell}$ the set containing all subsets of $\Sigma$ of size $\ell$, and for any pair of strings $x \in \Sigma^n$
and $S \in  {\Sigma \choose \ell}^n$ we denote by $\dist(x,S)$ the fraction of coordinates $i \in [n]$ for which $x_i \notin S_i$, that is,
$\dist(x,S):= \left|\left\{ i\in\left[n\right]:x_{i}\notin S_i \right\} \right|/n$. Throughout the paper, we use $\exp(n)$ to denote $2^{\Theta(n)}$. Whenever we use $\log$, it is to the base $2$.

\subsection{Error-correcting codes}

Let $\Sigma$ be an alphabet and $k, n$ be positive integers (the \textsf{message length} and the
\textsf{block length}, respectively). A \textsf{code} is an injective map $C: \Sigma^k \to \Sigma^n$. The elements in the domain of $C$ are called \textsf{messages} and the elements in the image of $C$ are called \textsf{codewords}.   If $\F$ is a finite field and $\Sigma$ is a vector space over $\F$,
we say that $C$ is \textsf{$\F$-linear}
if it is a linear transformation over $\F$ between the $\F$-vector spaces $\Sigma^k$ and $\Sigma^n$.
If $\Sigma=\F$ and $C$ is $\F$-linear, we simply say that $C$ is \textsf{linear}. The \textsf{generating matrix} of a linear code $C: \F^k \to \F^n$ is the matrix $G \in \F^{n \times k}$ such that $C(x) = G\cdot x$ for any $x\in \F^k$.
We say that a code $C: \Sigma^k \to \Sigma^n$ is \textsf{systematic} if any message is the prefix of its image, that is, for
any $x \in \Sigma^k$ there exists $y \in \Sigma^{n-k}$ such that $C(x)=(x,y)$.

The \textsf{rate} of a code $C: \Sigma^k \to \Sigma^n$ is the ratio $\rho := \frac{k} {n}$.
The \textsf{relative distance} $\dist(C)$ of $C$ is the minimum $\delta >0$ such that for
every pair of distinct messages $x,y\in \Sigma^k$ it holds that
$\dist(C(x),C(y))\ge\delta$.
For a code $C: \Sigma^k \to \Sigma^n$ of relative
distance $\delta$, a given parameter $\alpha<\delta/2$, and a string
$w\in\Sigma^n$, the \textsf{problem of decoding from $\alpha$~fraction
of errors} is the task of finding the unique message $x\in \Sigma^k$ (if any) which
satisfies $\dist(C(x),w)\leq\alpha$.

The best known general trade-off between rate and distance of codes is the \textsf{Gilbert-Varshamov bound}, attained by random (linear) codes. For $x \in [0,1]$ let
$$
H_q(x)=x\log_q(q-1)+x\log_q(1/ x) +(1-x)\log_q(1/(1-x))
$$
 denote the $q$-ary entropy function.
\begin{theorem}[Gilbert-Varshamov (GV) bound,  \cite{G52,V57}]\label{thm:GV}
\label{Gilbert-Varshamov} For any prime power $q$, $0\leq \delta <1-\frac 1 q$, $0\leq \rho < 1-H_q(\delta)$, and sufficiently large $n$, a random linear code
$C: \F_q^{\rho n} \to \F_q^{n}$ of rate $\rho$ has relative distance at least $\delta$ with probability at least $1 - \exp(-n)$.
\end{theorem}

\subsection{List decodable and list recoverable codes}

List decoding is a paradigm that allows one to correct more than $\delta/2$ fraction of errors by returning a small list of messages that correspond to close-by codewords.
More formally, for $\alpha \in [0,1]$ and an integer $L$ we say that a code $C: \Sigma^k \to \Sigma^n$ is \textsf{$(\alpha,L)$-list decodable} if for any $w \in \Sigma^n$ there are at most $L$ different messages $ x\in \Sigma^k$ which satisfy that $\dist(C(x),w)\leq \alpha$.

For list decoding concatenated codes it is useful to consider the notion of \textsf{list recovery} where one is given as input a small list of candidate symbols for each of the codeword coordinates, and is required to output a list of messages such that the corresponding codewords are consistent with the input lists.
More concretely, for $\alpha \in [0,1]$ and integers $\ell, L$ we say that a code $C: \Sigma^k \to \Sigma^n$ is \textsf{$(\alpha,\ell,L)$-list recoverable} if for any $S \in {\Sigma \choose \ell}^{n}$ there are at most $L$ different messages $x \in \Sigma^k$ which satisfy that $\dist(C(x),S)\leq \alpha$.

It is well-known that $1 - H_q(\alpha)$ is the \textsf{list decoding capacity}, that is, any $q$-ary code of rate above $1 - H_q(\alpha)$ cannot be list decoded from $\alpha$ fraction of errors with list size polynomial in the block length, and on the other hand, a random $q$-ary (linear) code of rate below $1 - H_q(\alpha)$ can be list decoded from $\alpha$ fraction of errors with small  list size.
\begin{theorem}[\cite{venkat-thesis}, Theorem 5.3]\label{thm:cap-list-dec}
For any prime power $q$, $0\leq \alpha < 1-\frac 1 q$,  $$0\leq \rho < 1 -H_q(\alpha) - \frac{1} {\log_q(L+1)},$$ and sufficiently large $n$, a random linear code
$C: \F_q^{\rho n} \to \F_q^{n}$  of rate $\rho$ is $(\alpha,L)$-list decodable with probability at least $1 - \exp(-n)$.
\end{theorem}

The following is a generalization of the above theorem to the setting of list recovery.
\begin{theorem}[\cite{venkat-thesis}, Lemma 9.6]
For any prime power $q$, integers $1\leq \ell \leq q$ and $L>\ell$, $0 \leq \alpha \leq 1$,
$$
0\leq  \rho <  \frac{1} {\log q}  \cdot \bigg[(1-\alpha)\cdot \log(q/\ell)- H_2(\alpha)- H_2(\ell/q)   \cdot \frac {q} {\log_q(L+1)}\bigg],
$$
and sufficiently large $n$, a random linear code
$C: \F_q^{\rho n} \to \F_q^{n}$ of rate $\rho$ is $(\alpha,\ell,L)$-list recoverable with probability at least $1 - \exp(-n)$.
\end{theorem}

Over large alphabet the above theorem yields the following.
\begin{corollary}\label{cor:list-rec-random}
There is a constant $c$ so that the following holds.
Choose $\rho \in [0,1]$, $ \eps > 0$, and a positive integer $\ell$.  Suppose that $q$ is a prime power which satisfies
 $$q \geq \max \{(1-\rho-\epsilon)^{-c(1-\rho-\epsilon) /\epsilon}, (\rho+\epsilon)^{-c(\rho+\epsilon)/\epsilon},  \ell^{c/\eps}\}.$$Then for sufficiently large $n$, a random linear code $C: \F_q^{\rho n} \to \F_q^{n}$  of rate $\rho$ is $(1-\rho -\epsilon,\ell, q^{c \ell/\epsilon})$-list recoverable with probability at least $1 - \exp(-n)$.
\end{corollary}

\begin{proof}
Follows by observing that in the above setting of parameters,
\begin{eqnarray*}
&& \frac{1} {\log q} \cdot  \bigg[(\rho+\epsilon)\cdot \log(q/\ell)- H_2(1-\rho-\epsilon)- H_2(\ell/q)   \cdot \frac {q} {c\ell/\epsilon}\bigg] \\
& \geq &\rho+\epsilon-\frac {\log \ell} {\log q}  - \frac {(1-\rho-\epsilon) \log(1/(1-\rho-\epsilon))} {\log q}- \frac {(\rho+\epsilon) \log(1/(\rho+\epsilon))} {\log q} -O(\epsilon / c)\\
&  \geq & \rho +\epsilon - O(\epsilon / c),
 \end{eqnarray*}
 so the corollary holds for sufficiently large constant $c$.
\end{proof}

\subsection{Locally decodable codes}

Intuitively, a code $C$ is said to be \textsf{locally decodable}
 if, given a codeword $C(x)$ that has been corrupted by some errors,
it is possible to decode any coordinate of the corresponding message $x$ by reading only a
small part of the corrupted version of $C(x)$. Formally, it is defined
as follows.
\begin{definition}[Locally decodable code (LDC)]
We say that a code $C: \Sigma^k \to \Sigma^{n}$ is $(Q,\alpha)$-\textsf{locally decodable}
if there
exists a randomized algorithm $A$ that satisfies the following requirements:
\begin{itemize}
\item \textbf{Input:} $A$ takes as input a coordinate $i\in\left[k\right]$,
and also gets oracle access to a string $ w\in\Sigma^{n}$ that is
$\alpha$-close to some codeword $C(x)$.
\item \textbf{Query complexity:} $A$ makes at most $Q$ queries to the
oracle $w$.
\item \textbf{Output:} $A$ outputs $x_{i}$ with probability at least $\frac{2}{3}$.
\end{itemize}
\end{definition}
\begin{remark}
By definition it holds that $\alpha<\dist(C)/2$. The above success
probability of~$\frac{2}{3}$ can be amplified using sequential repetition,
at the cost of increasing the query complexity. Specifically, amplifying
the success probability to $1-e^{-t}$ requires increasing the query
complexity by a multiplicative factor of $O(t)$.
\end{remark}

\paragraph{Locally list decodable and list recoverable codes.}

The following definition generalizes the notion of locally decodable codes to the setting of list decoding / recovery.
In this setting the algorithm $A$ is required to find all messages that correspond to nearby codewords in an implicit sense.

\begin{definition}[Locally list recoverable code]\label{defn:local-list-recover}
We say that a code $C: \Sigma^k \to \Sigma^{n}$ is $(Q,\alpha,\ell,L)$-\textsf{locally list recoverable}
if there exists a randomized algorithm $A$ that satisfies the following requirements:
\begin{itemize}
\item \textbf{Preprocessing:} $A$ outputs $L$ randomized algorithms $A_1,\ldots,A_L$.
\item \textbf{Input:} Each $A_j$ takes as input a coordinate $i\in\left[k\right]$,
and also gets oracle access to a string $S \in {\Sigma \choose \ell}^{n}$.
\item \textbf{Query complexity:} Each $A_j$ makes at most $Q$ queries to the
oracle $S$.
\item \textbf{Output:} For every codeword $C(x)$ that is $\alpha$-close to $S$,  with probability at least $\frac 2 3$ over the randomness of $A$ the following event happens:
there exists some $j\in [L]$ such that for all $i \in [k]$,  $$ \Pr[A_j(i) = x_i] \geq \frac{2}{3},$$ where the probability is over the internal randomness of $A_j$.
\end{itemize}
\end{definition}
We say that $A$ has \textsf{preprocessing time $T_{\text{pre}}$} if $A$ outputs the description of the algorithms $A_1,\ldots,A_L$ in time at most $T_{\text{pre}}$, and has \textsf{running time $T$} if each $A_j$ has running time at most $T$.
Finally, we say that $C$ is \textsf{$(Q,\alpha,L)$-locally list decodable} if it is $(Q,\alpha,1,L)$-locally list recoverable.

\subsection{Tensor codes}\label{sec:tensor}

A main ingredient in our constructions is the tensor product operation, defined as follows.
\begin{definition}[Tensor codes]\label{defn:tensor}
Let $C_{1}: \F^{k_1} \to \F^{n_{1}}$, $C_{2}: \F^{k_2} \to \F^{n_2}$ be linear codes, and let $G_1\in \F^{n_1 \times k_1}, G_2 \in \F^{n_2 \times k_2}$ be the generating matrices of $C_1, C_2$ respectively. Then the \textsf{tensor code} $C_{1} \otimes C_{2}: \F^{k_1 \times k_2} \to \F^{n_{1}\times n_{2}}$ is defined as $(C_1 \otimes C_2)(M) = G_1 \cdot M \cdot G_2^T$.
\end{definition}
Note that the codewords of $C_1 \otimes C_2$ are $n_1 \times n_2$ matrices over $\F$ whose columns belong to the code $C_1$ and whose rows belong to the code $C_2$.
The following effects of the tensor product operation on the classical parameters of the code  are well known (see e.g. \cite{S01,DSW06}).
\begin{fact}
Suppose that $C_{1}: \F^{k_1} \to \F^{n_{1}}$, $C_{2}: \F^{k_2} \to \F^{n_2}$ are linear codes of rates $\rho_{1},\rho_{2}$
and relative distances $\delta_{1},\delta_{2}$ respectively. Then the tensor code $C_1 \otimes C_2$ has rate $\rho_{1}\cdot \rho_{2}$ and relative distance $\delta_{1}\cdot\delta_{2}$.
\end{fact}
For a linear code $C$, let $C^{\otimes 1}:= C$ and $C^{\otimes t}:= C \otimes C^{\otimes (t-1)}$. Then by the above, if $C$ has rate $\rho$ and relative distance $\delta$ then $C^{\otimes t}$ has rate $\rho^t$ and relative distance $\delta^t$.

\section{High-rate locally list recoverable codes}\label{sec:locallistrecovery}

We start by showing the existence of high-rate locally list recoverable codes. For this we first show in Section \ref{subsec:approx-local-list-rec}
 below that high-rate tensor codes are {\em approximately locally list recoverable}, namely there exists a short list of local algorithms that can  recover {\em most} of the coordinates of messages that correspond to near-by codewords. We then observe in Section \ref{subsec:local-list-rec}
 that by pre-encoding the message with a locally decodable code, the former codes can be turned into {\em locally  list recoverable codes} for which the local algorithms can recover {\em all} the coordinates of messages that correspond to near-by codewords. Finally, we show in Section \ref{subsec:instant} how to instantiate the codes used in the process in order to obtain high-rate locally list recoverable codes with good performance.

\subsection{Approximate local list recovery}\label{subsec:approx-local-list-rec}

We start by showing that high-rate tensor codes are approximately locally list recoverable as per Definition \ref{defn:approx-local-list-recover} below. As noted above, the main difference between approximately locally list recoverable codes and locally list recoverable codes (Definition \ref{defn:local-list-recover}) is that in the former we only require that the local algorithms recover {\em most} of the coordinates. A simple averaging argument then shows that in the case of approximate local list recovery each of the local algorithms in the output list can be assumed to be {\em deterministic}. Finally, to describe our approximate local list recovery algorithm it will be more convenient to require that the local algorithms recover {\em codeword bits} as opposed to {\em message bits}\footnote{In our constructions we shall use  systematic codes and so recovery of codeword bits will imply also recovery of message bits.}.

\begin{definition}[Approximately locally list recoverable code]\label{defn:approx-local-list-recover}
We say that a code $C: \Sigma^k \to \Sigma^{n}$ is  $(Q,\alpha,\epsilon,\ell,L)$-\textsf{approximately locally list recoverable}
if there exists a randomized algorithm $A$ that satisfies the following requirements:
\begin{itemize}
\item \textbf{Preprocessing:} $A$ outputs $L$  deterministic algorithms $A_1,\ldots,A_L$.
\item \textbf{Input:} Each $A_j$ takes as input a coordinate $i\in\left[n\right]$,
and also gets oracle access to a string $S \in {\Sigma \choose \ell}^{n}$.
\item \textbf{Query complexity:} Each $A_j$ makes at most $Q$ queries to the
oracle $S$.
\item \textbf{Output:} For every codeword $C(x)$ that is $\alpha$-close to $S$,  with probability at least $1-\epsilon$ over the randomness of $A$ the following event happens:
there exists some $j\in [L]$ such that  $$ \Pr_{i \in [n]}[A_j(i) = C(x)_i] \geq 1-\epsilon,$$ where the probability is over the choice of uniform random $i\in [n]$.
\end{itemize}
\end{definition}
As before, we say that $A$ has \textsf{preprocessing time $\Tpre$} if $A$ outputs the description of the algorithms $A_1,\ldots,A_L$ in time at most $\Tpre$, and has \textsf{running time $T$} if each $A_j$ has running time at most $T$.

Our main technical lemma is the following.

\begin{lemma}[Main technical]\label{approx-local-list-rec} Suppose that $C: \F^k \to  \F^{n}$
is a linear code of relative distance $\delta$ that is $(\alpha,\ell,L)$-(globally)
list recoverable. Then for any $\widetilde \epsilon >0$, the tensor product code $\widetilde{C}:=C^{\otimes t}: \F^{k^t} \to  \F^{n^{t}}$
is $(\widetilde{Q},\widetilde{\alpha},\widetilde{\epsilon},\ell,\widetilde{L})$-approximately
locally list recoverable for $\widetilde{\alpha}=\alpha\cdot\widetilde{\epsilon}\cdot\delta^{O(t)}$,
$$
\widetilde{Q}=n \cdot \frac{\log^{t}L}{(\alpha\cdot\widetilde{\epsilon})^{O(t)}\cdot\delta^{O(t^{2})}},
$$
and
$$
\widetilde{L}=\exp\left(\frac{\log^{t}L}{(\alpha\cdot\widetilde{\epsilon})^{O(t)}\cdot\delta^{O(t^{2})}}\right).
$$
Moreover, the approximate local list recovery algorithm
		for $\widetilde{C}$ has preprocessing time
$$ \Ttpre = \log n \cdot  \exp\left(\frac{\log^{t}L}{(\alpha\cdot\widetilde{\epsilon})^{O(t)}\cdot\delta^{O(t^{2})}}\right)
,$$
		and if the (global) list recovery algorithm for $C$ runs in
		time $T$ then the approximate local list recovery algorithm for $\widetilde C$	runs in time
		$$
\widetilde{T}=T\cdot\frac{\log^{t}L}{(\alpha\cdot\widetilde{\epsilon})^{O(t)}\cdot\delta^{O(t^{2})}}.
$$
\end{lemma}

The proof of the above lemma will follow from repeated application of the following technical lemma.

\begin{lemma}\label{approx-local-list-rec-2-wise}
Suppose that $C: \F^k \to \F^{n}$ is a linear code of relative distance $\delta$ that is $(\alpha,\ell,L)$-(globally)
list recoverable, and $C': \F^{k'} \to \F^{n'}$ is a linear code that is $(Q',\alpha',\epsilon',\ell,L')$-approximately locally list
recoverable. 
Then for any  $\widetilde \epsilon \geq 100\epsilon'/\delta $, the tensor product code  $\widetilde{C}:=C \otimes C': \F^{k \times k'} \to \F^{n\times n'}$
is $(\widetilde{Q},\widetilde{\alpha},\widetilde{\epsilon},\ell,\widetilde{L})$-approximately
locally list recoverable for $\widetilde{\alpha}=\frac{1}{10}\cdot\min\left\{ \alpha'\cdot\delta,\alpha\cdot\widetilde{\epsilon}\right\} $,
		$$
		\widetilde{Q}=O\left(\frac{\log(L/\widetilde{\epsilon})}{(\delta\cdot\alpha'\cdot\widetilde{\epsilon})^{2}}\right)\cdot Q'+n,
		$$
		and
		$$
		\widetilde{L}=\exp\left(\frac{\log L'\cdot\log(L/\widetilde{\epsilon})}{(\delta\cdot\alpha'\cdot\widetilde{\epsilon})^{2}}\right).
		$$
		Moreover, if the (global) list recovery algorithm for $C$ runs in
		time $T$, and the approximate local list recovery algorithm for $C'$ has preprocessing time $\Tpre^\prime$ and
		runs in time $T'$, then the approximate local list recovery algorithm
		for $\widetilde{C}$ has preprocessing time
$$ \Ttpre = O\left(\frac{\log(L/\widetilde{\epsilon})}{(\delta\cdot\alpha'\cdot\widetilde{\epsilon})^{2}}\right) \cdot (\log n+\Tpre^\prime) + \exp\left(\frac{\log L'\cdot\log(L/\widetilde{\epsilon})}{(\delta\cdot\alpha'\cdot\widetilde{\epsilon})^{2}}\right),$$
		and	runs in time
		$$
		\widetilde{T}=O\left(\frac{\log(L/\widetilde{\epsilon})}{(\delta\cdot\alpha'\cdot\widetilde{\epsilon})^{2}}\right)\cdot T'+T.
		$$
		\end{lemma}

		Before we prove the above lemma we show how it implies the Main Technical Lemma \ref{approx-local-list-rec}.

		\begin{proof}[Proof of Lemma \ref{approx-local-list-rec}]
		The proof proceeds by repeated application of Lemma \ref{approx-local-list-rec-2-wise} to the code $C$.
		For a fixed $\widetilde{\eps} > 0$, our goal is to find parameters $\subp{Q}{t}$, $\subp{\alpha}{t}$, $\subp{L}{t}$ such that
		$C^{\otimes t}$ is a $(\subp{Q}{t},\subp{\alpha}{t},\widetilde{\eps},\ell,\subp{L}{t})$-approximately locally list recoverable code.

		We begin by defining
		\[
			\subp{\eps}{i} := \inparen{ \frac{\delta}{100} }^{t-i} \widetilde{\eps} \qquad \qquad \mbox{ for $i=1,\ldots,t$ },
		\]
		where the factor $\frac{\delta}{100}$ comes from Lemma \ref{approx-local-list-rec-2-wise}.
		With this definition, we have $\widetilde{\eps} = \subp{\eps}{t} > \subp{\eps}{t-1} > \cdots > \subp{\eps}{1}$,
		and
		\[
			\widetilde{\eps} =\subp{\epsilon}{t} =\inparen{ \frac{100}{\delta} }^{t-1} \subp{\eps}{1}.
		\]

		When $t=1$, we have $C^{\otimes t} = C$, which means $\subp{Q}{1} = n$, $\subp{\alpha}{1} = \alpha$, $\subp{L}{1} = L$, and
		this holds for any $\subp{\eps}{1} > 0$, and hence for any $\widetilde{\eps} > 0$, since $C$ is actually list recoverable (not just approximately).

		Since $C^{\otimes i} = C \otimes C^{\otimes (i-1)}$, by Lemma \ref{approx-local-list-rec-2-wise}, we have the following recursive relationships

		\begin{align*}
			\subp{\alpha}{i} &= \frac{1}{10} \min( \subp{\alpha}{i-1} \delta, \alpha \subp{\eps}{i} ) \\
			\subp{m}{i} &:= \frac{ \log(L/\subp{\eps}{i}) }{\inparen{ \delta \subp{\alpha}{i-1} \subp{\eps}{i} }^2} \\
			\subp{Q}{i} &= \subp{m}{i} \subp{Q}{i-1} + n\\
			\subp{L}{i} &= \subp{L}{i-1}^{\subp{m}{i}} \\
			\subp{T}{i} &= \subp{m}{i}\cdot \subp{T}{i-1} + T \\
			\subp{\inparen{\Tpre}}{i} &= \subp{m}{i}\cdot \inparen{ \log n + \subp{\inparen{\Tpre}}{i-1} } +\subp{L}{i-1}^{\subp{m}{i}}
		\end{align*}

		Solving these recursions gives the following bounds on the parameters of interest.
		The distance parameter, $\subp{\alpha}{t}$, satisfies
		\begin{align*}
			\subp{\alpha}{t} 	&= \frac{1}{10} \min( \subp{\alpha}{t-1} \delta, \alpha \subp{\eps}{t} ) \\
								&\ge \frac{1}{10} \min( \subp{\alpha}{t-1} \delta, \alpha \subp{\eps}{1} ) \\
								&\ge \inparen{ \frac{\delta}{10}}^{t-1} \alpha \subp{\eps}{1} \\
								&= \inparen{ \frac{\delta}{10}}^{t-1} \alpha \inparen{ \frac{\delta}{100} }^{t-1} \widetilde{\eps} \\
								&= \inparen{ \frac{\delta^2}{1000}}^{t-1} \alpha \widetilde{\eps}.
		\end{align*}
		
		Notice that as $i$ increases, the parameters $\subp{\eps}{i}$ are increasing, whereas the parameters $\subp{\alpha}{i}$ are decreasing.
		Nevertheless, $\subp{\alpha}{i} \subp{\eps}{i} > \subp{\alpha}{i-1}\subp{\eps}{i-1}$.
		To see this note that $\subp{\alpha}{i} < \frac{1}{10} \min( \subp{\alpha}{i-1}, \subp{\eps}{i} ) < \eps_{(i)}$ for all $i$,
		which means
		\[
			\frac{100}{\delta} \subp{\alpha}{i} = \frac{10}{\delta} \min( \subp{\alpha}{i-1} \delta, \subp{\eps}{i} ) > \subp{\alpha}{i-1},
		\]
		where the last inequality follows since $\alpha_{i-1} < \epsilon_{i-1} < \epsilon_{i}$.
		Multiplying both sides by $\subp{\eps}{i}$, we conclude that $\subp{\alpha}{i} \subp{\eps}{i+1} > \subp{\alpha}{i-1}\subp{\eps}{i}$, as desired.
		

		Now, notice that since $\subp{\eps}{i-1} < \subp{\eps}{i}$ and $\subp{\alpha}{i} \subp{\eps}{i+1} > \subp{\alpha}{i-1} \subp{\eps}{i}$, we have
		\[
			\subp{m}{i-1} = \frac{ \log(L/\subp{\eps}{i-1}) }{\inparen{ \delta \subp{\alpha}{i-2} \subp{\eps}{i-1} }^2}
							\ge \frac{ \log(L/\subp{\eps}{i}) }{\inparen{ \delta \subp{\alpha}{i-1} \subp{\eps}{i} }^2}
							= \subp{m}{i}
		\]
		Thus
		\[
			\subp{m}{2} \ge \subp{m}{3} \ge \cdots \ge \subp{m}{t}
		\]
		Notice that
		\[
			\subp{m}{2} = \frac{ \log(L/\subp{\eps}{2}) }{\inparen{ \delta \subp{\alpha}{1} \subp{\eps}{2} }^2} = \frac{ \log( \inparen{\frac{100}{\delta}}^{t-2} L/\subp{\eps}{t}) }{\inparen{ \delta \alpha \inparen{ \frac{\delta }{100} }^{t-2}\subp{\eps}{t} }^2} \le \frac{ \log( \inparen{\frac{100}{\delta}}^{t-2} L/\subp{\eps}{t}) }{ \inparen{ \frac{\delta }{100} }^{2t} \alpha^2 \subp{\eps}{t}^2} \le \frac{ \log(L/\subp{\eps}{t}) }{ \inparen{ \frac{\delta }{100} }^{3t} \alpha^2 \subp{\eps}{t}^2}
		\]
		Where last inequality holds since $\log(a x) \le a \log(x)$ when $\log(x) \ge \log(a)/(a-1)$.


		The output list size, $\subp{L}{t}$, satisfies
			\begin{align*}
			\subp{L}{t} &= L^{\prod_{i=2}^{t-1} \subp{m}{i}} \\
				&\le L^{\subp{m}{2}^{t-1}} \\
				&= \exp\inparen{( \log L )\cdot \subp{m}{2}^{t-1} } \\
				&\le \exp \inparen{ \frac{ \log^t(L/\widetilde{\eps})}{ \inparen{\frac{\delta}{100}}^{3t^2} \alpha^{2t} \widetilde{\eps}^{2t} } }.
				\end{align*}
		The query complexity, $\subp{Q}{t}$, satisfies
				\begin{align*}
				\subp{Q}{t} &=n \inparen{1 + \sum_{i=2}^{t-1} \prod_{j=i}^{t-1} \subp{m}{j} } \\
					&\le (t-1) n \subp{m}{2}^{t-1} 
					\\
					&\le n \inparen{ \frac{ \log^t(L/\widetilde{\eps})}{ \inparen{\frac{\delta}{100}}^{3t^2} \alpha^{2t} \widetilde{\eps}^{2t} } }.
\end{align*}

		The running time, $\subp{T}{t}$, satisfies
				\begin{align*}
				\subp{T}{t} &= T \prod_{i=2}^{t-1} \subp{m}{i} + T \inparen{1 + \sum_{i=2}^{t-1} \prod_{j=i}^{t-1} \subp{m}{j} } \\
					&\le T \inparen{ \frac{ \log^t(L/\widetilde{\eps})}{ \inparen{\frac{\delta}{100}}^{3t^2} \alpha^{2t} \widetilde{\eps}^{2t} } }.
				\end{align*}
		The pre-processing time, $\subp{\inparen{\Tpre}}{t}$, satisfies
			\begin{align*}
				\subp{\inparen{\Tpre}}{t} 	&= \subp{m}{t}\cdot \inparen{ \log n + \subp{\inparen{\Tpre}}{t-1} } + \exp \inparen{ \subp{m}{t} \log \subp{L}{t-1} } \\
											&= \log n \cdot \inparen{1+\sum_{i=2}^{t} \prod_{j=i}^{t-1} \subp{m}{j}} + \sum_{i=1}^{t-1} \subp{L}{i} \prod_{j=i}^{t-1} \subp{m}{j} \\
											&\le  \subp{m}{2}^{t-1} t \log n + t \subp{m}{2}^{t-1} \subp{L}{t} \\
											&\le (t \log n + t\subp{L}{t} ) \subp{m}{2}^{t-1} \\
											&\le \inparen{ t \log n + t \exp \inparen{ \frac{ \log^t(L/\widetilde{\eps})}{ \inparen{\frac{\delta}{100}}^{3t^2} \alpha^{2t} \widetilde{\eps}^{2t} } } }\frac{ \log^t(L/\widetilde{\eps})}{ \inparen{\frac{\delta}{100}}^{3t^2} \alpha^{2t} \widetilde{\eps}^{2t} } \\
											&\leq  \log n \cdot  \exp\left(\frac{\log^{t}L}{(\alpha\cdot\widetilde{\epsilon})^{O(t)}\cdot\delta^{O(t^{2})}}\right)
			\end{align*}

\end{proof}

We proceed to the proof of Lemma \ref{approx-local-list-rec-2-wise}.

\begin{proof}[Proof of Lemma \ref{approx-local-list-rec-2-wise}]
Our goal is to find a randomized algorithm $\widetilde{A}$ that outputs a list of (deterministic) local algorithms $\widetilde{A}_{1},\ldots, \widetilde{A}_{\widetilde L}$ such that for any codeword $\widetilde c \in C\otimes C'$ that is consistent with most of the input lists, with high probability over the randomness of $A$, there exists some $\widetilde{A}_i$ in the output list that computes correctly most of the coordinates of $\widetilde c$.

We first describe the algorithm $\widetilde{A}$.
The algorithm $\widetilde{A}$ first
chooses a uniform random subset $R\subseteq[n]$ of rows of size $m:=O\left(\frac{\log(L/\widetilde{\epsilon})}{(\delta\cdot\alpha'\cdot\widetilde{\epsilon})^{2}}\right)$.
It then runs for each row $r\in R$, independently, the approximate local
list recovery algorithm $A'$ for $C'$, let $A_{1}^{r},\ldots,A_{L'}^{r}$
denote the output algorithms on row $r$. Finally,
for every possible choice of a single local algorithm $A_{j_r}^{r}$ per each of the rows $r\in R$, the algorithm
$\widetilde{A}$ outputs a local algorithm denoted $\widetilde{A}_{J}$ where
$J:=(j_{r})_{r\in R}\in[L']^{R}$. The formal definition of the algorithm $\widetilde{A}_{J}$  is given below, followed by an informal description. \\

\begin{algorithm}
\begin{algorithmic}
\Function{$\widetilde{A}_J$}{$(i,i') \in [n] \times [n']$}
\LineComment{$\widetilde{A}_J$ receives oracle access to a matrix of lists $S\in{\F \choose \ell}^{n\times n'}$}
\LineComment{ $J = (j_r)_{r \in R} \in [L']^R$ }
\For{ $r \in R$ }
	\State{ Run $A_{j_r}^{r}$ on input $i'$ and oracle access to the $r$th row $S|_{\{r\} \times [n']}$ .}
	\State{ Let $c'_{r} \leftarrow A^r_{j_r}(i')$. }
	\LineComment{ $c'_{r}$ is a candidate for the symbol at position $(r,i') \in [n] \times [n']$. }
\EndFor
\LineComment{ At this point, we have candidate symbols for every position in $R \times \{i'\}$. }
\State{ Run the (global) list recovery algorithm for $C$ on the $i'$th column $S|_{[n] \times \{i'\}}$. }
\State{ Let $\mathcal{L}\subseteq \F^{n}$ denote the output list. }
\State{ Choose a codeword $c\in\mathcal{L}$ such that $c|_{R}$ is closest to $(c'_r)_{r\in R}$ (breaking ties arbitrarily). }
\State{\textbf{Return:} $c_i$}
\EndFunction
\end{algorithmic}
\caption{ The approximate local list recovery algorithm  for $C \otimes C'$. \label{alg:allr} }
\end{algorithm}

Recall, that the algorithm $\widetilde{A}_J$ is given as input a codeword coordinate $(i,i') \in [n] \times [n']$ in the tensor product code $C \otimes  C'$, is
allowed to query the input lists at every coordinate, and must produce a guess for the codeword value indexed by $(i,i')$.   To this end, the algorithm $\widetilde{A}_J$ first runs on each row $r \in R$ the local recovery algorithm $A_{j_r}^{r}$ for $C'$ that is specified by the choice of $J=(j_r)_{r \in R}$ on input $i'$ and oracle access to  the $r$th row. For each row $r \in R$ let $c'_r$ be the guess for
the symbol at position $(r,i') \in [n] \times [n']$ produced by $A_{j_r}^{r}$. At this point we have candidate symbols for all positions in $R \times \{i'\}$. Now, $\widetilde{A}_J$ runs the global list recovery algorithm for $C$ on the $i'$th column and chooses a codeword $c \in C$ from the output list that agrees the most with the candidate symbols $(c'_r)_{r\in R}$ on this column. Finally,  the $i$th symbol of $c$ is $\widetilde{A}_J$'s guess for the $(i,i')$ symbol of the tensor codeword. Next we prove the correctness of the purposed local list recovery algorithm, followed by an analysis of its performance.

\paragraph{Correctness:}
Let $\widetilde{c}$ be a codeword of $\widetilde C = C \otimes C'$ such that $\dist(\widetilde{c},S)\leq\widetilde{\alpha}$. Our goal is to show that with probability at least $1-\widetilde \epsilon$ over the randomness of $\widetilde{A}$, there exists some local algorithm $\widetilde{A}_{J}$ in the output list of $\widetilde A$ that computes correctly at least $1- \widetilde \epsilon$ fraction of the coordinates of $\widetilde{c}$.

We first explain how the algorithm $\widetilde{A}_{J}$ above is obtained.
Recall that $A_{j_r}^r$ is the local algorithm that computes the $j_r$th ``guess'' for the codeword on the $r$th row.
For every row $r\in R$ let $j_{r}\in[L']$ be such that $A_{j_{r}}^{r}$
agrees the most with
 $\widetilde{c}$ on the $r$th row among all local algorithms $A_1^r,\ldots,A_{L'}^r$
  (breaking ties arbitrarily), and let $J=(j_{r})_{r\in R}$.
We shall show below that with probability at least $1-\widetilde{\epsilon}$
over the randomness of $\widetilde{A}$, it holds that $\widetilde{A}_{J}$
computes correctly at least $1-\widetilde{\epsilon}$ fraction of the coordinates of $\widetilde{c}$.

The high level idea of the proof is as follows.
First, we observe that since the rows in $R$ are chosen uniformly at random, and by averaging, with high probability for most rows $r \in R$ it holds that $\widetilde c$ is consistent with most of the input lists on the row. Let us denote by  $R_{\text{good}} \subseteq R$  the subset of these 'good' rows. By the local list recovery guarantee for $C'$, with high probability on each such good row $r \in R_{\text{good}}$ the algorithm $A_{j_r}^{r}$ computes correctly most of the coordinates of $\widetilde c$ on this row. Now, by another averaging argument this implies in turn that for most columns $i' \in [n']$ it holds that both $\widetilde c$ is consistent with most of the input lists on the columnn and additionally, most of the guesses $(c'_r)_{r\in R}$ for this column are correct. As above, let us denote by $\text{Col}_{\text{good}} \subseteq [n']$ the subset of these 'good' columns. Finally, by the list recovery guarantee for $C$, on any good column $i' \in \text{Col}_{\text{good}}$ the $i'$th column $\widetilde c|_{[n] \times \{i'\}}$ of $\widetilde c$ is present in the output list $\mathcal{L}$, and we further show that with high probability $\widetilde c |_{[n] \times \{i'\}}$ is closest to $(c'_r)_{r\in R}$ on $R$ among all codewords in $
\mathcal{L}$, in which case $\widetilde{A}_{J}$ outputs the correct $(i,i')$ symbol of $\widetilde c$. Details follow.

We start by showing the existence of a large subset of good rows $R_{\text{good}} \subseteq R$.
For this observe that since the rows in $R$ are chosen uniformly at random, and since $\widetilde{\alpha}\leq\frac{1}{10}\cdot\alpha'\cdot\delta$,
by Chernoff bound (without replacement, see e.g. Lemma 5.1 in \cite{GGR}),
with probability at least $1-\exp(-(\alpha'\cdot\delta)^{2}m)\geq1-\frac{\widetilde{\epsilon}}{3}$
over the choice of $R$, it holds that $\dist(\widetilde{c}|_{R\times[n']},S|_{R\times[n']})\leq\frac{\alpha'\cdot\delta}{8}$.
If this is the case, then by averaging, for at least $1-\frac{\delta}{8}$ fraction of the rows
$r\in R$ it holds that $\dist\left(\widetilde{c}|_{\{r\}\times[n']},S|_{\{r\}\times[n']}\right)\leq\alpha'$.
Let
$$
R_{\text{good}}=\left\{ r\in R\mid \dist\left(\widetilde{c}|_{\{r\}\times[n']},S|_{\{r\}\times[n']}\right)\leq\alpha'\right\} .
$$

Now we observe that for each good row $r\in R_{\text{good}}$, with high probability over the randomness of $A'$ the algorithm $A_{j_r}^{r}$ computes correctly most of the coordinates of $\widetilde c$ on this row.
For this note that since $1 - \eps' \ge 1- \frac{1}{100} \cdot \delta \cdot \widetilde{\epsilon}$, with probability at least
$1- \frac{1}{100} \cdot \delta \cdot \widetilde{\epsilon}$ over the randomness of $A'$, independently for each row $r$, it holds that $A_{j_{r}}^{r}$
computes correctly at least $1-\frac{1}{100}\cdot\delta\cdot\widetilde{\epsilon}$ fraction
of the coordinates of $\widetilde{c}|_{\{r\}\times[n']}$. Consequently, by Chernoff bound
with probability at least $1-\exp(-(\delta\cdot\widetilde{\epsilon})^{2}\cdot m)\geq1-\frac{\widetilde{\epsilon}}{3}$
over the randomness of $\widetilde{A}$, it holds that $(A_{j_{r}}^{r}){}_{r\in R_{\text{good}}}$
compute correctly at least $1-\frac{1}{24}\cdot\delta\cdot\widetilde{\epsilon}$ fraction
of the coordinates of $\widetilde{c}|_{R_{\text{good}}\times[n']}$ . Finally, if this is the case then for at least $1-\frac{\widetilde{\epsilon}}{3}$
fraction of the columns $i'\in[n']$ it holds that $\dist\left(\widetilde{c}|_{R_{\text{good}}\times\{i'\}},(A_{j_{r}}^{r}(i')){}_{r\in R_{\text{good}}}\right)\leq\frac{\delta}{8}$.

Next we show the existence of a large subset of good columns $\text{Col}_{\text{good}}\subseteq [n']$.
So far we obtained that with probability at least $1-\frac{2}{3}\cdot\widetilde{\epsilon}$
over the randomness of $\widetilde{A}$, for at least $1-\frac{\widetilde{\epsilon}}{3}$ fraction
of the columns $i'\in[n']$ it holds that $\dist\left(\widetilde{c}|_{R\times\{i'\}},(A_{j_{r}}^{r}(i'))_{r\in R}\right)\leq\frac{\delta}{4}$.
Moreover, note that since $\widetilde{\alpha}\leq\frac{1}{10}\cdot\alpha\cdot\widetilde{\epsilon}$,
for at least $1-\frac{\widetilde{\epsilon}}{3}$ fraction of the columns
$i'\in[n']$ it holds that $\dist(\widetilde{c}|_{[n]\times\{i'\}},S|_{[n]\times\{i'\}})\leq\alpha$.
Let
$$
\text{Col}_{\text{good}}=\left\{ i'\in[n']\mid \dist(\widetilde{c}|_{[n]\times\{i'\}},S|_{[n]\times\{i'\}})\leq\alpha\;\text{{and}\;}\dist\left(\widetilde{c}|_{R\times\{i'\}},(A_{j_{r}}^{r}(i'))_{r\in R}\right)\leq\frac{\delta}{4}\right\} .
$$
Then by the above, with probability at least $1-\frac{2}{3}\cdot\widetilde{\epsilon}$
over the randomness of $\widetilde{A}$ it holds that $|\text{Col}_{\text{good}}|\geq(1-\frac{2}{3}\cdot\widetilde{\epsilon})|n'|$.

Now  for each column $i' \in \text{Col}_{\text{good}}$ it holds that $\dist(\widetilde{c}|_{[n]\times\{i'\}},S|_{[n]\times\{i'\}})\leq\alpha$, and so $\widetilde{c}|_{[n]\times\{i'\}}$ must be in the output list $\mathcal{L}$ of the $i'$th column. Moreover, since the code $C$ has relative distance $\delta$, any codeword
$\hat c \in \mathcal{L}$ other than $\widetilde{c}|_{[n]\times\{i'\}} $ must differ from  $\widetilde{c}|_{[n]\times\{i'\}}$ by at least $\delta$ fraction of the coordinates. Furthermore, since $R$ is chosen uniformly at random, by Chernoff bound, with probability at least $1-\exp(-\delta^{2}m)$ over the choice of $R$ it holds that $\dist(\hat c|_R, \widetilde{c}|_{R\times\{i'\}}) \geq 3\delta/4$. By union over all codewords in $\mathcal{L}$ this implies in turn that  with probability at least $1-L\cdot \exp(-\delta^{2}m)\geq1-\frac{1}{9}\cdot\widetilde{\epsilon}^{2}$
over the choice of $R$, for all codewords  $\hat c\in \mathcal{L} \setminus \{\widetilde{c}|_{[n]\times\{i'\}}\}$  it holds that
 $\dist(\hat c|_R,\widetilde{c}|_{R\times\{i'\}} ) \geq 3\delta/4$.

 Next observe that if the above holds then for any column $i' \in \text{Col}_{\text{good}}$
 we have on the one hand that
 $ \dist\left(\widetilde{c}|_{R\times\{i'\}},(c'_r)_{r\in R}\right)\leq\frac{\delta}{4}$, and on the other hand that
 $$ \dist\left(\hat c|_R, (c'_r)_{r\in R}\right)\geq \dist(\hat c|_R, \widetilde{c}|_{R\times\{i'\}}) -\dist\left(\widetilde{c}|_{R\times\{i'\}},(c'_r)_{r\in R}\right) \geq \frac{\delta}{2}$$
  for
 any $\hat c  \in \mathcal{L} \setminus \{\widetilde{c}|_{[n]\times\{i'\}}\}$.
  So $\widetilde{c}|_{[n]\times\{i'\}}$ will be the codeword in $\mathcal{L}$ that is closest to $(c'_r)_{r\in R}$ on $R$, and consequently we will have $c=\widetilde{c}|_{[n]\times\{i'\}}$ and $c_i = \widetilde c_{i,i'}$.
 So we obtained that for each column $i' \in \text{Col}_{\text{good}}$,
  with probability at least $1-\frac{1}{9}\cdot\widetilde{\epsilon}^{2}$ over the randomness of $\widetilde A$,
the algorithm $\widetilde A_{J}$ computes the entire column $i'$ correctly.  By averaging, this implies in turn that with probability at least $1-\frac{\widetilde{\epsilon}}{3}$  over the randomness of $\widetilde A$ the algorithm $\widetilde A_J$ computes correctly at least
$1-\frac{\widetilde{\epsilon}}{3}$  fraction of the coordinates of $\widetilde c|_{[n] \times \text{Col}_{\text{good}}}$.

Concluding, we obtained that with probability at least $1-\widetilde{\epsilon}$
over the randomness of $\widetilde{A}$, the algorithm $\widetilde A_{J}$ computes
correctly at least $1-\widetilde{\epsilon}$ fraction of the coordinates of $\widetilde{c}$, so the algorithm $\widetilde A$ satistifes the local list recovery requirement. Next we analyze the performance of the algorithm.


\paragraph{Output list size:}
The resulting output list size equals the number of different strings $(j_r)_{r \in R} \in [L']^R$ which is
\[
\widetilde{L}=(L')^{m}=\exp\left(\frac{\log(L')\cdot\log(L/\widetilde{\epsilon})}{(\delta\cdot\alpha'\cdot\widetilde{\epsilon})^{2}}\right).
\]

\paragraph{Query complexity:}
The query complexity is
\[
\widetilde{Q}=m\cdot Q'+n=O\left(\frac{\log(L/\widetilde{\epsilon})}{(\delta\cdot\alpha'\cdot\widetilde{\epsilon})^{2}}\right)\cdot Q'+n,
\]
since $Q'$ queries are needed in order to list recover each of the rows in $R$, and $n$ additional queries are needed to globally list recover the $i'$'th column.

	\paragraph{Pre-Processing Time:}
		The pre-processing algorithm generates a random set $R$, of size $m$, which takes $m \log n$ time to generate and store.
		The pre-processing algorithm then runs $m$ independent copies of $A^\prime$ (once for each row in $R$), and this takes time $m \cdot \Tpre^\prime$.
		Finally, the pre-processing algorithm generates the set $J$ of size $\inparen{L^\prime}^m$.
		Thus the total pre-processing time is
		\[
			m \inparen{ \log n + \Tpre^\prime } + \exp \inparen{ m \log L^\prime }.
		\]
	\paragraph{Running Time:}
		Each local recovery algorithm runs $m$ copies of the local recovery algorithm for $C^\prime$, which takes time $m \cdot T^\prime$.
		Then it runs the global list recovery algorithm for $C$ once, which takes time $T$, thus
		the total running time of each local recovery algorithm is
		\[
			m \cdot T^\prime + T.
		\]

\end{proof}

\subsection{Local list recovery}\label{subsec:local-list-rec}

Next we show that the approximately locally list recoverable codes of Lemma \ref{approx-local-list-rec} can be turned into locally list recoverable codes by pre-encoding the message with a locally decodable code.

\begin{lemma}\label{lem:local-list-rec}
Suppose that $C: \F^{\rho n} \to \F^{n}$ is a {\em systematic} linear code of  rate $\rho$ and relative distance $\delta$ that is $(\alpha,\ell,L)$-(globally)
list recoverable, and $\widehat C: \F^{\widehat k} \to \F^{(\rho n)^t}$ is $(\widehat Q,\widehat \alpha)$-locally decodable.
Then $\widetilde{C}:=C^{\otimes t}\big(\widehat C\big): \F^{\widehat k} \to \F^{n^t} $ is $(\widetilde{Q},\widetilde{\alpha},\ell,\widetilde{L})$-locally
list recoverable for $\widetilde{\alpha}=\alpha\cdot\widehat\alpha\cdot\rho^t\cdot \delta^{O(t)}$,
$$
\widetilde{Q}=\widehat Q\cdot n \cdot\frac{\log^{t}L}{(\alpha\cdot \widehat \alpha)^{O(t)}\cdot\left(\rho\cdot\delta\right)^{O(t^{2})}},
$$
and
$$
\widetilde{L}=\exp\left(\frac{\log^{t}L}{(\alpha\cdot\widehat \alpha)^{O(t)}\cdot\left(\rho\cdot\delta\right)^{O(t^{2})}}\right).
$$
Moreover, the local list recovery algorithm for $\widetilde{C}$ has preprocessing time
$$\Ttpre=\exp\left(\frac{\log^{t}L}{(\alpha\cdot\widehat \alpha)^{O(t)}\cdot\left(\rho\cdot\delta\right)^{O(t^{2})}}\right) \cdot \log n ,$$ and
if the (global) list recovery algorithm for $C$ runs in
time $T$ and the local decoding algorithm for $\widehat C$ runs in time
$\widehat T$ then the local list recovery algorithm for $\widetilde{C}$ runs
in time
$$
\widetilde{T} =\widehat T+ \widehat Q \cdot T\cdot\frac{\log^{t}L}{(\alpha\cdot\widehat \alpha)^{O(t)}\cdot\left(\rho\cdot\delta\right)^{O(t^{2})}}.
$$
\end{lemma}

\begin{proof}
	Setting $\widetilde{\eps} = \widehat{\alpha} \rho^t$ in Lemma \ref{approx-local-list-rec}, we conclude that the tensor code $\overline{C} := C^{\otimes t} : \F^{(\rho n)^t} \rightarrow \F^{n^t}$ is
	$(\overline{Q},\overline{\alpha},\widetilde{\eps},\ell,\overline{L})$-approximately locally list recoverable for $\overline{\alpha} = \alpha \widetilde{\eps} \delta^{O(t)}$. Note furthermore that since $C$ is systematic then so is $\overline C$.

	Intuitively, the proof works as follows: to recover the $i$th message symbol, $x_i$, run the local decoder of the inner code $\widehat{C}$ to obtain a set of $\widehat{Q}$ indices
	in $\F^{(\rho n)^t}$ that, if queried, would allow you to recover $x_i$.  Since the code $\overline{C}$ is systematic, those symbols correspond to symbols in the big code $\widetilde{C}$.
	Use the approximate local list recovery algorithm for $\overline{C}$ to obtain $\overline{L}$ guesses for each of these $\widehat{Q}$ symbols.  Finally, for each of these $\overline{L}$
	sets of $\widehat{Q}$ ``guesses'' run the local decoding algorithm for $\widehat{C}$ to obtain $\overline{L}$ guesses for $x_i$.
	Since $\overline{C}$ is only approximately locally list recoverable, there will be a subset of symbols on which the approximate local list decoder fails, but by
	carefully choosing parameters, these errors can be handled by the local decoding procedure of $\widehat{C}$.

	It is not hard to see that the query complexity of this algorithm will be $\widehat{Q}$ times the query complexity of $C^{\otimes t}$, and the output list size
	will be the same as that of $C^{\otimes t}$.

	Below, we describe the algorithm and proof of correctness in more detail.
	Let $(\overline{A}_1,\cdots,\overline{A}_{\overline{L}}) \leftarrow \overline{A}(\cdot)$ be the approximate local list recovery algorithms for $\overline{C}$.
	In Algorithm \ref{alg:llr} we describe the local list recovery algorithms $\widetilde{A}_1,\ldots,\widetilde{A}_{\overline{L}}$ for the code $\widetilde{C} \defined \overline{C}(\widehat{C})$.
	
	\begin{algorithm}
		\begin{algorithmic}
		\Function{$\widetilde{A}_j$}{$i \in [\widehat{k}]$}
		\LineComment{$\widetilde{A}_j$ receives oracle access to lists $S\in{\F \choose \ell}^{n^t}$}
		\State{Run the local decoding algorithm for $\widehat{C}$ on input $i$ to obtain a set of $\widehat{Q}$ indices that the local decoder would query. }
		\State{Let $R \subseteq [(\rho n)^t]$ be the subset of indices that would be queried. }
		\State{Let $\overline{R} \subseteq [n^t]$ be the indices in $\overline{C}$ encoding the indices of $R$. }
		\LineComment{$\overline{R}$ exists and $|\overline{R}| = |R| = \widehat{Q}$ because $\overline{C}$ is systematic.}
		\For{ $\overline{r} \in \overline{R}$ }
			\State{ Let $c^{(\overline{r})}_j \gets \overline{A}_j(\overline{r})$ (on oracle access to $S$) }
		\EndFor
		\State{ Run the local decoder for $\widehat{C}$ on input $\{ c^{(\overline{r})}_j \}_{\overline{r} \in \overline{R}}$ to obtain a guess $x^{(j)}_i$ for the $i$th symbol of the message}
		\State{\textbf{Return:} $x^{(j)}_i$}
		\EndFunction
		\end{algorithmic}
		\caption{The local list recovery algorithm for $\widetilde{C} \defined C^{\otimes{t}}(\widehat{C})$. \label{alg:llr}}
	\end{algorithm}

			\paragraph{Correctness:}
			Given a string of lists $S\in{\F \choose \ell}^{n^t}$, suppose $\widetilde{c} = \widetilde{C}(x) = \overline{C}(\widehat{C}(x))$ is a codeword of $\widetilde{C}$ that is $\widetilde{\alpha}$-close to $S$.
			We need to show that with probability at least $\frac{2}{3}$ over the randomness of $\widetilde{A}$ there is a $j \in [\overline{L}]$ such that
			for all $i \in [\widehat k]$, $\Pr[ \widetilde{A}_j(i) = x_i ] \ge \frac{2}{3}$.

			Since $\widetilde{\eps} = \widehat{\alpha}\rho^t$ the tensor code $\overline{C} = C^{\otimes t}$ is $(\overline{Q},\overline{\alpha},\widetilde{\eps},\ell,\overline{L})$-approximately locally list recoverable,
			for $\overline{\alpha} = \alpha \widetilde{\eps} \delta^{O(t)}$.  Thus if $\widetilde{c}$ is $\widetilde{\alpha}$-close to $S$,
			then with probability $1-\widetilde{\eps}$, over the randomness of $\overline{A}$, there is a $j \in [\overline{L}]$ such that
			\[
				\Pr_{i \in [n^t]} [\overline A_j(i) = \overline{C}(x)_i ] \geq 1 - \widetilde{\eps}.
			\]
			This means that with probability $1 - \widetilde{\eps}$ over the randomness of $\overline{A}$, there is a $j$ such that
			$\overline{A}_j(\cdot)$ is $\widetilde{\eps}$-close to the codeword $\overline{C}(\widehat{C}(x))$.
			Since $\overline{C}$ is systematic, the inner codeword $\widehat{C}(x)$ must
			be $\widetilde{\eps} \rho^{-t}$-close to the restriction of $S$ to the information symbols.
			Since $\widehat{C}$ is $(\widehat{Q},\widehat{\alpha})$-locally decodable, and $\widehat{\alpha} = \widetilde{\eps} \rho^{-t}$,
			then by the local decoding property of the inner code, $\widehat{C}$, given any $i \in [\widehat{k}]$,
			and oracle access to $\overline{A}_j(\cdot)$, the local decoding algorithm for $\widehat{C}$ will recover $x_i$
			with probability at least $\frac{2}{3}$.
	
			\paragraph{Output list size and query complexity:}
			The query complexity is
			\[
				\widetilde{Q} = \widehat{Q} \cdot \overline{Q} = \widehat Q \cdot n\cdot \frac{\log^{t}L}{(\alpha\cdot\widehat{\alpha} \cdot \rho^t)^{O(t)}\cdot\delta^{O(t^{2})}},
			\]
			and the output list size is
			\[
				\widetilde{L} = \overline{L} = \exp\left(\frac{\log^{t}L}{(\alpha\cdot \widehat{\alpha} \cdot \rho^t)^{O(t)}\cdot\delta^{O(t^{2})}}\right).
			\]

	It can also be verified that the running times are as required.

\end{proof}

\subsection{Instantiations}\label{subsec:instant}

In what follows we shall instantiate Lemma~\ref{lem:local-list-rec} in two ways. For both, we will use the high-rate LDCs of~\cite{KMRS15-STOC} (Theorem \ref{thm:high-rate-lccs} below) as the code $\widehat C$.
In the first instantiation, which is more straightforward, we just use a random linear code (via Corollary \ref{cor:list-rec-random}) as the code $C$.  This yields a code $\widetilde{C}$ that is not efficiently encodable, nor efficiently list recoverable, but it does achieve small locality together with constant alphabet size and constant output list size.  The second instantiation, which does yield efficinetly encodable (in nearly-linear time) and efficiently list recoverable (in sub-linear time) codes, uses a modification of the Algebraic Geometry subcodes studied in~\cite{GX13, GK14} (Theorem \ref{thm:ag}) as the code $C$.  These latter codes have constant alphabet size, but slightly super-constant output list size, which means that our efficient construction will as well. In more detail, we obtain the following pair of theorems.

\begin{theorem}[High-rate locally list recoverable codes, non-efficient]\label{thm:nonexplicit-high-rate-local-list-rec}
There is a constant $c$ so that the following holds.
Choose $\eps > 0$ and positive integers $\ell,t$.  Suppose that $s \geq \max \{1/\epsilon^c,   c  (\log \ell) t  /\eps\} $.

Then there exists  an infinite family of $\F_{2}$-linear codes $\left\{ C_{n}\right\} _{n}$
such that the following holds.
\begin{enumerate}
\item $C_n: \F_{2^s}^{(1-\epsilon)n} \to \F_{2^s}^n$ has rate $1 - \eps$ and relative distance at least $(\epsilon/(16 t))^{t}$.
\item $C_n$ is $(Q,\alpha, \ell, L)$-locally list recoverable for $\alpha = (\epsilon/t)^{O(t)}$,
$$ Q  =n^{1/t} \cdot 2^{O(\sqrt{\log n \cdot \log \log n})} \cdot (s\ell)^t \cdot (t/\epsilon)^{O(t^2)}  ,$$
and
$$L = \exp \left((s \ell)^{t} \cdot (t/\epsilon)^{O(t^2)}  \right).$$
\end{enumerate}
In particular, when $\epsilon, \ell, t, s$ are constant we get that  $\alpha = \Omega(1)$, $Q = n^{1/t+o(1)}$, and $L=O(1)$.
\end{theorem}

\begin{theorem}[High-rate locally list recoverable codes, efficient]\label{thm:explicit-high-rate-local-list-rec}
There is a constant $c$ so that the following holds.
Choose $\eps > 0$ and a positive integer $\ell$.
Let $\left\{t_n\right\}_n$ be a sequence of positive integers, non-decreasing with $n$, so that $t_0$ is sufficiently large and
\[ t_n \leq \sqrt{ \frac{\eps \log_q(n) }{ c \ell }}.\]  For each choice of $t$ choose
	$s = s(t)$ so that
   $s \geq \max \{1/\epsilon^c, c (\log \ell) t  /\eps\} $ is even.
Then there exists  an infinite family of $\F_{2}$-linear codes $\left\{ C_{n}\right\} _{n}$
such that the following holds.
Below, to simplify notation we use $t$ instead of $t_n$ and $s$ instead of $s(t_n)$.
\begin{enumerate}

\item $C_n: \F_{2^s}^{(1-\epsilon)n} \to \F_{2^s}^n$ has rate $1 - \eps$ and relative distance at least $(\Omega({\epsilon} /{t}))^{2t}$.
\item $C_n$ is $(Q,\alpha, \ell, L)$-locally list recoverable  for $\alpha = (\epsilon/t)^{O(t)}$,
$$ Q = n^{1/t} \cdot 2^{O(\sqrt{\log n \cdot \log \log n})} \cdot  \exp\left(   \frac{ t^2 \ell s } {\eps} \cdot \exp(\log^*n) +t\log s\right),$$
and
$$L = \exp \left( \exp\left(   \frac{ t^2 \ell s} {\eps} \cdot \exp(\log^*n) +t \log s\right) \right).$$

\item The local list recovery algorithm for $C_n$ has preprocessing time
$$ \Tpre = \exp \left( \exp\left(   \frac{ t^2 \ell s } {\eps} \cdot \exp(\log^*n) +t\log s\right) \right) \cdot \log n ,$$
and running time
$$ T =    n^{O(1/t)} \cdot 2^{O(\sqrt{\log n \cdot \log \log n})} \cdot  \exp \left( \exp\left(   \frac{ t \ell s } {\eps} \cdot \exp(\log^*n) +\log s\right) \right) .$$
\item $C_n$ can be encoded in time
$$ n \cdot 2^{O(\sqrt{\log n \cdot \log \log n})} +  t \cdot  n^{1+O(1/t)} .$$
\end{enumerate}

In particular, when $\epsilon, \ell, t_n = t, s$ are constant we get that  $\alpha = \Omega(1)$, $Q = n^{1/t+o(1)}$,
$L=  \exp(\exp(\exp(\log^*n)))$, $\Tpre  = \log^{1+o(1)} n$, $T  = n^{O(1/t)}$, and encoding time is $n^{1+O(1/t)}$.
\end{theorem}

\begin{remark}[Super-constant $t$]
	Theorem~\ref{thm:explicit-high-rate-local-list-rec} is interesting even when $t_n = t$ is a sufficiently large constant that does not depend on $n$.  For our applications, we will need to take $t_n$ to be slightly super-constant, so we allow for that in the statement of Theorem~\ref{thm:explicit-high-rate-local-list-rec}.
	\end{remark}

To prove the above theorems we use the following result from~\cite{KMRS15-STOC, KMRS15_subpoly} that shows the existence of high-rate LDCs with sub-polynomial query complexity.

\begin{theorem}[\cite{KMRS15_subpoly}, Theorem 1.3]\label{thm:high-rate-lccs}
There is a constant $c$ so that the following holds.
Choose $\rho \in [0,1]$ and $\varepsilon>0$.
Then there exists an infinite sequence $\{n_i\}_i$ such that for any $n=n_i$ in the sequence and for any $s \geq 1/\epsilon^c$ there exists an
$\F_2$-linear code $C_n$ satisfying:
\begin{enumerate}
\item $C_{n}: \F_{2^s}^{\rho n} \to \F_{2^s}^n$ has  rate $\rho$ and relative distance
at least $1-\rho-\varepsilon$.
\item $C_{n}$ is $(2^{O(\sqrt{\log n\cdot\log\log n})},\frac{1-\rho-\varepsilon}{2})$-locally decodable in time
$s\cdot 2^{O(\sqrt{\log n\cdot\log\log n})}$.
\item $C_n$ can be encoded in time $n \cdot 2^{O(\sqrt{\log n\cdot\log\log n})}$.
\end{enumerate}
\end{theorem}

\begin{remark} We remark about a few differences between the above theorem and Theorem 1.3 in \cite{KMRS15_subpoly}:
\begin{enumerate}
\item Theorem 1.3 in \cite{KMRS15_subpoly} talks about locally correctable codes (LCCs) instead of locally decodable codes (LDCs). The difference between LCCs and LDCs is that for LCCs the local correction algorithm is required to decode codeword bits as opposed to message bits. However, it can be shown that the result holds for LDCs as well (see discussion at end of Section 1.1 in \cite{KMRS15_subpoly}).
\item Theorem 1.3 in \cite{KMRS15_subpoly} only states the existence of {\em some} $s_0 \leq 1/\varepsilon^c$ for which the above holds, however it can be verified that the result holds for {\em any} $s \geq s_0$ as well.
\item Encoding time is not stated explicitly in  \cite{KMRS15_subpoly}, Lemma 3.2 and  \cite{Kop15}, Appendix A.
\end{enumerate}
\end{remark}

We proceed to the proof of Theorem \ref{thm:nonexplicit-high-rate-local-list-rec}.

\begin{proof}[Proof of theorem \ref{thm:nonexplicit-high-rate-local-list-rec}]
Let $c'$ be a sufficiently large constant for which both Corollary \ref{cor:list-rec-random} and Theorem \ref{thm:high-rate-lccs} hold, and suppose that $s \geq \max \{(4/\epsilon)^{c'}, 16 {c'} (\log \ell)  t  /\eps\}$.  

Let $C: \F_{2^s}^{((1-\epsilon/2)n)^{1/t}} \to \F_{2^s}^{n^{1/t}}$ be a linear code of rate $(1- \epsilon/2)^{1/t} \leq 1- \epsilon/(8t)$
(the inequality holds since $(1-x)^y \leq 1 - xy/4$ for $x,y \in [0,1]$, see e.g. Fact 2.1 in \cite{KMRS15_subpoly})
that is $( \frac {\epsilon} {16 t} , \ell, 2^{O(s\ell t/\epsilon)})$-list recoverable
whose existence is guaranteed by Corollary \ref{cor:list-rec-random} for sufficiently large $n$ (depending on $\epsilon, \ell, t,s$). Note furthermore that by Theorem \ref{thm:GV} we may assume that $C$ has relative distance at least
$  {\epsilon} /{(16 t)}$. Finally, note that
one may assume that the code $C$ is systematic by performing Gaussian elimination on the generating matrix of $C$.
Let $\widehat C: \F_{2^s}^{(1-\epsilon)n} \to \F_{2^s}^{(1-\epsilon/2)n}$ be an $\F_2$-linear code of rate
$\frac{1-\epsilon} {1-\epsilon/2} \leq 1- \frac \epsilon 2$ that is
$\big(2^{O(\sqrt{\log n \cdot \log \log n})} ,\frac \epsilon 4\big)$-locally decodable given by Theorem \ref{thm:high-rate-lccs} for infinite values of $n$ (depending on $\epsilon$).

Let $C_n:=C^{\otimes t}(\widehat C)$ for any $n$ for which both $C$ and $\hat C$ exist.
Then $C_n: \F_{2^s}^{(1-\epsilon)n} \to \F_{2^s}^{n}$  is an $\F_2$-linear code of rate $1-\epsilon$ and relative distance at least
$(\epsilon/(16 t))^{t}$.
Moreover, by Lemma \ref{lem:local-list-rec} the code $C_n$ is
$(Q,\alpha, \ell, L)$-locally list recoverable for $\alpha = (\epsilon/t)^{O(t)}$,
$$ Q = n^{1/t} \cdot 2^{O(\sqrt{\log n \cdot \log \log n})} \cdot \bigg(\frac{s \ell t} {\epsilon}\bigg)^{t} \cdot (t/\epsilon)^{O(t^2)} =n^{1/t} \cdot 2^{O(\sqrt{\log n \cdot \log \log n})} \cdot (s\ell)^t \cdot (t/\epsilon)^{O(t^2)} ,$$
and
$$L = \exp \left((s \ell)^{t} \cdot (t/\epsilon)^{O(t^2)} \right).$$
\end{proof}

Next we prove Theorem \ref{thm:explicit-high-rate-local-list-rec}.

\begin{proof}[Proof of Theorem \ref{thm:explicit-high-rate-local-list-rec}]

Fix any $n \in \mathbb{N}$
so that Theorem \ref{thm:high-rate-lccs} may be instantiated with block length $(1 - \eps/2)n$, and so that
\begin{equation}
\label{eq:chooset}
n^{1/t_n} \geq q^{8c_0\ell t_n/\epsilon},
\end{equation}
where $c_0$ is the constant from the statement of Theorem~\ref{thm:ag}.
By Theorem~\ref{thm:high-rate-lccs} and the assumption on $t_n$, there are infinitely many such $n$.
For the rest of the proof, we will denote $t_n$ by $t$, since $n$ is now fixed.

Let $c'$ be a sufficiently large constant (independent of $n$) for which both Theorem \ref{thm:ag} and Theorem \ref{thm:high-rate-lccs} hold, and suppose that $s \geq \max \{(4/\epsilon)^{c'}, 8 c'  (\log \ell) t  /\eps\}$ is even.
Then the code $C_n$ is constructed as follows.

Let  $C: \F_{2^s}^{((1-\epsilon/2)n)^{1/t}} \to \F_{2^s}^{n^{1/t}}$ be a linear code of rate $1 - \epsilon/(8t) \geq (1- \epsilon/2)^{1/t}$ and relative distance $\Omega(({\epsilon} /{8t})^2)$, that is $(\Omega(({\epsilon} /{8t})^2), \ell, L_0)$-list recoverable whose existence is guaranteed by Theorem \ref{thm:ag}, where
\begin{align*}
L_0 &= \exp_{2^s} \inparen{ \exp_{2^s} \inparen{ \frac{ \ell t }{\eps} \cdot \exp(\log^*(n))}  } \\
&= \exp\inparen{ \exp\inparen{ \frac{\ell s t}{\eps} \exp(\log^*(n)) } + \log(s)}.
\end{align*}
Here, we are using \eqref{eq:chooset} to ensure that we may choose the block length $n^{1/t}$ and rate $1 - \epsilon/(8t)$ in Theorem~\ref{thm:ag}.

Once more,
one may further assume that the code $C$ is systematic by performing Gaussian elimination on the generating matrix of $C$.
Let $\widehat C: \F_{2^s}^{(1-\epsilon)n} \to \F_{2^s}^{(1-\epsilon/2)n}$ be an $\F_2$-linear code of rate
$\frac{1-\epsilon} {1-\epsilon/2} \leq 1- \frac \epsilon 2$ that is
$\big( 2^{O(\sqrt{\log n \cdot \log \log n})}, \frac \epsilon 4\big)$-locally decodable given by Theorem \ref{thm:high-rate-lccs}.

Let $C_n:=C^{\otimes t}(\widehat C)$.
Then $C_n: \F_{2^s}^{(1-\epsilon)n} \to \F_{2^s}^{n}$  is an $\F_2$-linear code of rate $1-\epsilon$ and relative distance at least $(\Omega({\epsilon} /{t}))^{2t}$. Moreover, by Lemma \ref{lem:local-list-rec} the code $C_n$ is
$(Q,\alpha, \ell, L)$-locally list recoverable for $\alpha = (\epsilon/t)^{O(t)}$,
$$ Q = n^{1/t} \cdot 2^{O(\sqrt{\log n \cdot \log \log n})} \cdot  \exp\left(   \frac{ t^2 \ell s } {\eps} \cdot \exp(\log^*n) +t\log s\right),$$
and
$$L = \exp \left( \exp\left(   \frac{ t^2 \ell s } {\eps} \cdot \exp(\log^*n) +t\log s\right) \right).$$

Next observe that since $C$ can be list recovered in time $\poly(n^{1/t},L_0)$, and $\widehat C$ can be locally decoded in time $2^{O(\sqrt{\log n \cdot \log \log n})}$, the local list recovery algorithm for $C_n$
has preprocessing time
$$ \Tpre = \exp \left( \exp\left(   \frac{ t^2 \ell s } {\eps} \cdot \exp(\log^*n) +t\log s \right) \right) \cdot \log n ,$$
and running time
$$ T =    n^{O(1/t)} \cdot 2^{O(\sqrt{\log n \cdot \log \log n})} \cdot  \exp \left( \exp\left(   \frac{ t \ell s } {\eps} \cdot \exp(\log^*n) +\log s\right) \right) .$$

Finally, since $C$ can be encoded in time $\poly(n^{1/t})$
then $C^{\otimes t}$ can be encoded in time $t \cdot n \cdot \poly(n^{1/t})$,
and consequently the encoding time of $C_n:=C^{\otimes t}(\widehat C)$ is
$$
n \cdot 2^{O(\sqrt{\log n \cdot \log \log n})} + t \cdot n^{1+O(1/t)}.
$$
\end{proof}

\section{Capacity achieving locally list decodable codes}\label{sec:caplldc}
Next we show the existence of capacity achieving locally list decodable codes over large (but constant) alphabet. As before we exhibit two instantiations of this result: In the first instantiation we obtain capacity achieving locally list decodable codes that have constant output list size, however these codes cannot be efficiently encoded or list decoded;  In the second instantiation we obtain capacity achieving locally list decodable codes that are efficiently encodable (in nearly-linear time) and efficiently list decodable (in sub-linear time), however these codes have slightly super-constant output list size. These latter codes can also achieve $n^{o(1)}$ query complexity at the cost of increasing the alphabet and output list sizes to $n^{o(1)}$.
Towards our GV bound application, we actually present a stronger version that applies also to list recovery, where list decoding corresponds to the special case in which the input list size $\ell$ equals $1$.

\begin{theorem}[Capacity achieving locally list decodable / recoverable codes, non-efficient]\label{thm:nonexplicit-cap-local-list-rec}
There is a constant $c$ so that the following holds.
Choose $\rho \in [0,1]$, $\eps > 0$, a positive integer $\ell$, and sufficiently large integer $t$.  Let
$$s_0:= \max \big\{c (1-\rho-\epsilon/2) \log(1/(1-\rho-\epsilon/2)) /\epsilon,  c (\rho+\epsilon/2) \log(1/(\rho+\epsilon/2)) /\epsilon,c (\log\ell)/\eps\big\},$$ and suppose that  $s \geq  \max \big\{(t/\epsilon)^{ct},  c  \ell t /\eps^2\big\} \cdot s_0/\rho.$

Then there exists  an infinite family of $\F_{2}$-linear codes $\left\{ C_{n}\right\} _{n}$
such that the following holds.
\begin{enumerate}
\item $C_n: \F_{2^s}^{\rho n} \to \F_{2^s}^n$ has rate $\rho$ and relative distance at least $1-\rho-\epsilon$.
\item $C_n$ is $(Q,1-\rho-\epsilon, \ell, L)$-locally list recoverable for
$$ Q  =n^{1/t} \cdot 2^{O(\sqrt{\log n \cdot \log \log n})} \cdot s^{t+O(1)} \cdot 2^{O(s_0 \ell t /\epsilon)} \cdot (t/\epsilon)^{O(t^2)}  ,$$
and
$$L = \exp \left(s^t \cdot 2^{O(s_0 \ell t /\epsilon)} \cdot (t/\epsilon)^{O(t^2)}  \right) .$$
\end{enumerate}

In particular, when $\rho, \epsilon, \ell, t,s$ are constant we get that  $Q = n^{1/t+o(1)}$ and $L=O(1)$.
\end{theorem}

\begin{theorem}[Capacity achieving locally list decodable / recoverable codes, efficient]\label{thm:explicit-cap-local-list-rec}
There is a constant $c$ so that the following holds.
Choose $\rho \in [0,1]$, $\eps > 0$, and a positive integer $\ell$.
Let $\left\{t_n\right\}_n$ be a sequence of positive integers, non-decreasing with $n$, so that $t_0$ is sufficiently large, and so that
\[ t_n \leq \sqrt{ \frac{\eps \log_q(n) }{ c \ell }}.\]
Let
$$s_0:= \max \big\{c (1-\rho-\epsilon/2) \log(1/(1-\rho-\epsilon/2)) /\epsilon,  c (\rho+\epsilon/2) \log(1/(\rho+\epsilon/2)) /\epsilon,c (\log\ell)/\eps\big\},$$ and for each choice of $t$ choose $s = s(t)$ so that $s \geq    \max \big\{(t/\epsilon)^{ct},  c  \ell t /\eps^2\big\} \cdot s_0/\rho$ is even.

Then there exists  an infinite family of $\F_{2}$-linear codes $\left\{ C_{n}\right\} _{n}$
such that the following holds.
Below, we use $t$ to denote $t_n$ and $s$ to denote $s(t_n)$ to simplify notation.
\begin{enumerate}
\item $C_n: \F_{2^s}^{\rho n} \to \F_{2^s}^n$ has rate $\rho$ and relative distance at least $1-\rho-\epsilon$.
\item $C_n$ is $(Q,1-\rho-\epsilon, \ell, L)$-locally list recoverable for
$$ Q = n^{1/t} \cdot 2^{O(\sqrt{\log n \cdot \log \log n})} \cdot  \exp\left(  t^2 s  \cdot 2^{O(s_0 \ell /\epsilon)}   \cdot \exp(\log^*n) +t\log s\right)$$
and
$$L = \exp \left( \exp\left(   t^2 s \cdot 2^{O(s_0 \ell /\epsilon)} \cdot \exp(\log^*n) +t\log s\right) \right).$$

\item The local list recovery algorithm for $C_n$ has preprocessing time
$$ \Tpre = \exp \left( \exp\left(  t^2s  \cdot 2^{O(s_0 \ell /\epsilon)} \cdot \exp(\log^*n) +t\log s\right) \right) \cdot \log n ,$$
and running time
$$ T =    n^{O(1/t)} \cdot 2^{O(\sqrt{\log n \cdot \log \log n})} \cdot  \exp \left( \exp\left(   t s \cdot 2^{O(s_0 \ell /\epsilon)}  \cdot \exp(\log^*n) +\log s\right) \right) .$$
\item $C_n$ can be encoded in time
$$ n \cdot 2^{O(\sqrt{\log n \cdot \log \log n})} +  t \cdot  n^{1+O(1/t)} + O(n \cdot 2^{s^2})+s^{O(1)} \cdot n \cdot \polylog n.$$
\end{enumerate}

In particular,
\begin{itemize}
\item
When $\rho, \epsilon, \ell, t,s$ are constant we see that $Q = n^{1/t+o(1)}$,
$L=  \exp(\exp(\exp(\log^*n)))$, $T_{\text{pre}}=\log^{1+o(1)} n$, $T  = n^{O(1/t)}$, and encoding time is $n^{1+O(1/t)}$.
\item Alternatively, when $\rho, \epsilon, \ell$ are constant, 
$t = t_n = O\inparen{ \frac{\log\log\log(n)}{ (\log\log\log\log(n))^2}}$, 
and $s=t^{O(t)}$  we see that $Q,  L, T_{\text{pre}}, T$ are of the form $n^{o(1)}$ 
and encoding time is $n^{1+O(1/t)}=n^{1+o(1)}$.
\end{itemize}
\end{theorem}

To prove the above theorems we shall use
the following lemma from \cite{GKORS17} that gives a distance amplification procedure for local list recovery.

\begin{lemma}[\cite{GKORS-ECCC}, Lemma 5.4]\label{lem:AEL-list-rec}
There is a constant $c$ so that the following holds.
For any $\delta_{\out},\alpha_{\out},\epsilon>0$
there exists an integer
$d\leq(\delta_{\out}\cdot \alpha_{\out}\cdot \epsilon)^{-c}$
such that the following holds.
\begin{itemize}
\item Let $C_{\out}: (\Sigma_{\out})^{\rho_{\out} \cdot n_{\out}} \to (\Sigma_{\out})^{n_{\out}} $ be an $\F$-linear code of rate $\rho_{\out}$ and relative distance $\delta_{\out}$
that is $(Q,\alpha_{\out},\ell_{\out},L_{\out})$-locally list
recoverable.
\item Let $C_{\inn}: (\Sigma_{\inn})^{\rho_{\inn} \cdot n_{\inn}} \to (\Sigma_{\inn})^{n_{\inn}} $ be an $\F$-linear code of
rate $\rho_{\inn}$ and relative distance $\delta_{\inn}$ that is
$(\alpha_{\inn},\ell_{\inn},L_{\inn})$-(globally) list recoverable.
\item Additionally, suppose that $n_{\inn}\geq d$,
$|\Sigma_{\out}|=|\Sigma_{\inn}|^{\rho_{\inn}\cdot n_{\inn}}$ and $L_{\inn}\leq\ell_{\out}$.
\end{itemize}
Then there exists an $\F$-linear code $C:  (\Sigma_{\inn}^{n_{\inn}})^{(\rho_{\inn} \cdot \rho_{\out}) \cdot n_{\out}} \to  (\Sigma_{\inn}^{n_{\inn}})^{n_{\out}}$ of rate $\rho_{\inn}\cdot \rho_{\out}$ and relative
distance at least $\delta_{\inn}-\epsilon$ that is $(O( Q \cdot n_{\inn}^2\cdot\log n_{\inn}),\alpha_{\inn}-\epsilon,\ell_{\inn},L_{\out})$-locally
list recoverable.

Moreover,
\begin{itemize}
 \item If the  local list recovery algorithm for $C_{\out}$ has preprocessing time $T_{\text{pre},\out}$ and running time $T_{\out}$, and the global list recovery algorithm for $C_{\inn}$ has running time $T_{\inn}$, then  the local list recovery algorithm for $C$ has preprocessing time $T_{\text{pre},\out}$ and running time $$O(T_{\out}) + O(Q \cdot T_{\inn}) + \poly (Q, n_{\inn}, \ell_{\inn}). $$
 \item If the encoding times of $C_{\out},C_{\inn}$ are $\widehat T_{\out}, \widehat T_{\inn}$ respectively then the encoding time of $C$ is $$O(\widehat T_{\out}) + O(n_{\out} \cdot \widehat T_{\inn})+ n_{\out} \cdot \poly(n_{\inn},\log n_{\out}) .$$
 \end{itemize}
\end{lemma}

\begin{remark}
The definition of locally list recoverable codes in \cite{GKORS17} is stronger than ours since it requires an additional soundness property which guarantees that with high probability, all local algorithms in the output list compute an actual codeword. This requirement is not needed in our setting, and it can be verified that the proof of the above lemma goes through also with our weaker definition. Also, the encoding time was not mentioned explicitly in \cite{GKORS17} but it can be deduced from the proof of the lemma.
\end{remark}

We proceed to the proof of Theorem \ref{thm:nonexplicit-cap-local-list-rec}.

\begin{proof}[Proof of Theorem \ref{thm:nonexplicit-cap-local-list-rec}]

Let $c'$ be a sufficiently large constant for which both Corollary \ref{cor:list-rec-random}, Theorem \ref{thm:nonexplicit-high-rate-local-list-rec} and Lemma \ref{lem:AEL-list-rec} hold, and suppose that $s \geq     \max \big\{(t/\epsilon)^{c't},  32(c')^2  \ell t /\eps^2\big\} \cdot s_0/\rho$ for
$$s_0:= \max \big\{4c' (1-\rho-\epsilon/2) \log(1/(1-\rho-\epsilon/2)) /\epsilon,  4c' (\rho+\epsilon/2) \log(1/(\rho+\epsilon/2)) /\epsilon,4c' (\log\ell)/\eps\big\}.$$

Let $C_{\inn}: \F_{2^{s_0}}^{(\rho+\epsilon/4) \cdot (s/s_0)} \to \F_{2^{s_0}}^{s/s_0}$ be a linear code of rate $\rho+\epsilon/4$ that is $\big(1-\rho-\epsilon/2,\ell, 2^{ 4c' s_0 \ell/\epsilon}\big)$-list recoverable whose existence is guaranteed by Corollary \ref{cor:list-rec-random} for sufficiently large $t$ (depending on $\rho,\epsilon,\ell$).  Note furthermore that by Theorem \ref{thm:GV} we may assume that $C_{\inn}$ has relative distance at least $1-\rho-\epsilon/2$.

Let $C_{\out}: \F_{2^{(\rho+\epsilon/4) s}}^{\rho  n/(\rho +\epsilon/4)} \to \F_{2^{(\rho+\epsilon/4) s}}^n$ be an $\F_2$-linear code of rate $\frac{\rho} {\rho+\epsilon/4}\leq 1 -\epsilon/8$ and relative distance at least
$(\epsilon/(128  t))^{t}$, given by Theorem \ref{thm:nonexplicit-high-rate-local-list-rec} for infinite values of $n$ (depending on $\rho,\epsilon,\ell,t,s$),
 that is $(Q_{\out},(\epsilon/t)^{O(t)}, 2^{4c' s_0 \ell/\epsilon}, L_{\out})$-locally list recoverable for
$$ Q_{\out}  =n^{1/t} \cdot 2^{O(\sqrt{\log n \cdot \log \log n})} \cdot s^t \cdot 2^{O(s_0 \ell t /\epsilon)} \cdot (t  /\epsilon)^{O(t^2)}  $$
and
$$L_{\out} = \exp \left( s^t \cdot 2^{O(s_0 \ell t /\epsilon)} \cdot (t/\epsilon)^{O(t^2)}  \right).$$

Then by Lemma \ref{lem:AEL-list-rec} for any $n$ for which $C_{\out}$ exists there exists an $\F_2$-linear code $C_n: \F_{2^s}^{\rho n} \to \F_{2^s}^{n}$ of rate $\rho$
and relative distance at least $1-\rho-\epsilon$ that is $(Q,1-\rho-\epsilon, \ell, L)$-locally list recoverable for
$$ Q  =n^{1/t} \cdot 2^{O(\sqrt{\log n \cdot \log \log n})} \cdot s^{t+O(1)} \cdot 2^{O(s_0 \ell t /\epsilon)} \cdot (t/\epsilon)^{O(t^2)}  ,$$
and
$$L = \exp \left(s^t \cdot 2^{O(s_0 \ell t /\epsilon)} \cdot (t/\epsilon)^{O(t^2)}  \right) .$$

\end{proof}

Next we prove Theorem \ref{thm:explicit-cap-local-list-rec}.

\begin{proof}[Proof of theorem \ref{thm:explicit-cap-local-list-rec}]
Let $c'$ be a sufficiently large constant for which both Corollary \ref{cor:list-rec-random}, Theorem \ref{thm:explicit-high-rate-local-list-rec} and Lemma \ref{lem:AEL-list-rec} hold.
Fix $n \in \mathbb{N}$ so that Theorem~\ref{thm:explicit-high-rate-local-list-rec} guarantees the existence of a code with block length $n$ and using parameter $t = t_n$.  (By Theorem~\ref{thm:explicit-high-rate-local-list-rec}, there are infinitely many such $n$).  Now that $n$ is fixed, we will use $t$ to denote $t_n$.
Suppose that $s = s(t) \geq   \max \big\{(t/\epsilon)^{c't},  32(c')^2  \ell t /\eps^2\big\} \cdot s_0/\rho$ is even for
$$s_0:= \max \big\{4c' (1-\rho-\epsilon/2) \log(1/(1-\rho-\epsilon/2)) /\epsilon,  4c' (\rho+\epsilon/2) \log(1/(\rho+\epsilon/2)) /\epsilon,4c' (\log\ell)/\eps\big\}.$$
The code $C_n$ is constructed as follows.

Let $C_{\inn}: \F_{2^{s_0}}^{(\rho+\epsilon/4) \cdot (s/s_0)} \to \F_{2^{s_0}}^{s/s_0}$ be a linear code of rate $\rho+\epsilon/4$ that is $\big(1-\rho-\epsilon/2,\ell, 2^{ 4c' s_0 \ell/\epsilon}\big)$-list recoverable whose existence is guaranteed by Corollary \ref{cor:list-rec-random}.  Note furthermore that by Theorem \ref{thm:GV} we may assume that $C_{\inn}$ has relative distance at least $1-\rho-\epsilon/2$.

Let $C_{\out}: \F_{2^{(\rho+\epsilon/4) s}}^{\rho  n/(\rho +\epsilon/4)} \to \F_{2^{(\rho+\epsilon/4) s}}^n$ be an $\F_2$-linear code of rate $\frac{\rho} {\rho+\epsilon/4}\leq 1 -\epsilon/8$ and relative distance at least
$(\Omega(\epsilon/t))^{2t}$, given by Theorem \ref{thm:explicit-high-rate-local-list-rec}, that is $(Q_{\out},(\epsilon/t)^{O(t)}, 2^{4 c' s_0 \ell/\epsilon}, L_{\out})$-locally list recoverable for
$$ Q_{\out} = n^{1/t} \cdot 2^{O(\sqrt{\log n \cdot \log \log n})} \cdot  \exp\left(  t^2 s \cdot 2^{O(s_0 \ell /\epsilon)}   \cdot \exp(\log^*n) +t\log s\right)$$
and
$$L_{\out} = \exp \left( \exp\left(   t^2s \cdot 2^{O(s_0 \ell /\epsilon)} \cdot \exp(\log^*n) +t\log s\right) \right).$$

Then by Lemma \ref{lem:AEL-list-rec} there exists an $\F_2$-linear code $C_n: \F_{2^s}^{\rho n} \to \F_{2^s}^{n}$ of rate $\rho$
and relative distance at least $1-\rho-\epsilon$ that is $(Q,1-\rho-\epsilon, \ell, L)$-locally list recoverable for
$$ Q = n^{1/t} \cdot 2^{O(\sqrt{\log n \cdot \log \log n})} \cdot  \exp\left(  t^2 s  \cdot 2^{O(s_0 \ell /\epsilon)}   \cdot \exp(\log^*n) +t\log s\right)$$
and
$$L = \exp \left( \exp\left(   t^2 s  \cdot 2^{O(s_0 \ell /\epsilon)} \cdot \exp(\log^*n) +t\log s\right) \right).$$
The stated running and encoding times follow similarly.

\end{proof}

\section{Nearly-linear time capacity achieving list decodable codes}\label{sec:lintime}

In this section we show how to construct capacity achieving  list decodable codes that can be encoded and list decoded (probabilistically) in  nearly-linear time. These codes are presented in Section \ref{subsec:near-linear-cap-list-dec} below.
We then show (in Section \ref{subsec:near-linear-GV})
 how these codes can be used to probabilistically construct codes of rate up to $0.02$ that can be uniquely decoded (probabilistically) up to half the Gilbert-Varshamov (GV) bound in nearly-linear time.

\subsection{Nearly-linear time capacity achieving list decodable codes}\label{subsec:near-linear-cap-list-dec}

Our nearly-linear time capacity achieving list decodable codes follow as a consequence of our efficient capacity achieving locally list decodable code construction (Theorem \ref{thm:explicit-cap-local-list-rec}). Once more, we show a stronger version that applies also to list recovery.

\begin{theorem}[Nearly-linear time capacity achieving list decodable / recoverable codes]\label{thm:near-linear-cap-list-rec}
There is a constant $c$ so that the following holds.
Choose $\rho \in [0,1]$, $\eps > 0$, and a positive integer $\ell$.
Let $\left\{t_n\right\}_n$ be a sequence of positive integers, non-decreasing with $n$, so that $t_0$ is sufficiently large, and so that
\[ t_n \leq \sqrt{ \frac{\eps \log_q(n) }{ c \ell }}.\]
Let
$$s_0:= \max \big\{c (1-\rho-\epsilon/2) \log(1/(1-\rho-\epsilon/2)) /\epsilon,  c (\rho+\epsilon/2) \log(1/(\rho+\epsilon/2)) /\epsilon,c (\log\ell)/\eps\big\},$$ and for each choice of $t$, suppose that $s = s(t)$ is such that  $s \geq   \max \big\{(t/\epsilon)^{ct},  c  \ell t /\eps^2\big\} \cdot s_0/\rho$ is even.

Then there exists  an infinite family of $\F_{2}$-linear codes $\left\{ C_{n}\right\} _{n}$
such that the following holds.
Below, to simplify notation we use $t$ instead of $t_n$ and $s$ instead of $s(t_n)$.
\begin{enumerate}

\item $C_n: \F_{2^s}^{\rho n} \to \F_{2^s}^n$ has rate $\rho$ and relative distance at least $1-\rho-\epsilon$.
\item $C_n$ is $(1-\rho-\epsilon, \ell, L)$-list recoverable for
$$L = \exp \left( \exp\left(   t^2 s \cdot 2^{O(s_0 \ell /\epsilon)} \cdot \exp(\log^*n) +t\log s)\right) \right).$$
\item  $C_n$ can be list recovered probabilistically (with success probability $2/3$)\footnote{More precisely, there exists a randomized algorithm $A$, running in time $T$, that outputs a list of $L$ messages, and the guarantee is that with probability at least $2/3$ the output list  contains all messages that correspond to close-by codewords. Note that the success probability can be amplified to $1-e^{-t}$, at the cost of increasing the output list size by a multiplicative factor of $O(t)$, by repeating the algorithm independently $O(t)$ times and returning the union of all output lists.}  in time
$$T =    n^{1+O(1/t)} \cdot 2^{O(\sqrt{\log n \cdot \log \log n})} \cdot  \exp \left( \exp\left(   t^2 s \cdot 2^{O(s_0 \ell /\epsilon)}  \cdot \exp(\log^*n) +t\log s\right) \right) .  $$
\item $C_n$ can be encoded in time
$$ n \cdot 2^{O(\sqrt{\log n \cdot \log \log n})} +  t \cdot  n^{1+O(1/t)} + O(n \cdot 2^{s^2})+s^{O(1)} \cdot n \cdot \polylog n.$$
\end{enumerate}

In particular,
\begin{itemize}
\item
When $\rho, \epsilon, \ell, t,s$ are constant we get that
$L=  \exp(\exp(\exp(\log^*n)))$ and list recovery and encoding times are $n^{1+O(1/t)}$.
\item Alternatively, when $\rho, \epsilon, \ell$ are constant, $t = t_n =O \inparen{ \frac{\log\log\log(n)}{(\log\log\log\log(n))^2)}}$
and $s=t^{O(t)}$  we get that  $L=n^{o(1)}$ and list recovery and encoding times are $n^{1+O(1/t)}=n^{1+o(1)}$.
\end{itemize}
\end{theorem}

\begin{proof}
We use the same codes as in Theorem \ref{thm:explicit-cap-local-list-rec} so it only remains to prove the second and third bullets.

The global list recovery algorithm $A$ for $C_n$ first runs the local list recovery algorithm $A'$ for $C_n$ guaranteed by Theorem \ref{thm:explicit-cap-local-list-rec} for $O( \log L')$ times independently where $L'$ is the output list size given by Theorem \ref{thm:explicit-cap-local-list-rec}. Let $A'_1, A'_2,\ldots,A'_{O(L' \log L')}$ denote the local algorithms in the output lists. Then for each $A'_j$ the algorithm $A$ includes in the output list a message $x^{(j)} \in \F_{2^s}^{\rho n}$ which results by applying $A'_j$ on each of the $\rho n$ message coordinates $ O( \log (n L'))$ times independently and taking majority vote for each of the coordinates.

We claim that with probability at least $2/3$ the list output by $A$ includes all messages that correspond to close-by codewords. To see this note first that since any  close-by codeword is represented in the output list of $A'$ with probability at least $2/3$, it must hold that the number of close-by codewords is at most $3L'/2$. Consequently, running $A'$ for $O( \log L')$ times guarantees that with probability at least $5/6$ all  close-by codewords will be represented in the output list. Moreover, repeating the decoding of each message coordinate for $O( \log (n  L'))$ times guarantees that each of these coordinates is decoded correctly with probability at least $1- 1/(10  nL')$, and so by union bound with probability at least $5/6$ all messages are decoded correctly. So we obtained that with probability at least $1 - 1/6 -1/6 = 2/3$ the output list of $A$ will include all messages that correspond to close-by codewords.

Finally note that the output list size of the algorithm $A$ is
$$
O(L' \log L') = \exp \left( \exp\left(   t^2 s  \cdot 2^{O(s_0 \ell /\epsilon)} \cdot \exp(\log^*n) +t\log s\right) \right),
$$
 and its running time is
\begin{eqnarray*}
&& O(T'_{\text{pre}}\cdot \log L') +O( L'  \log L' \cdot n  \log (n L') \cdot T')  \\
 && = n^{1+O(1/t)} \cdot 2^{O(\sqrt{\log n \cdot \log \log n})} \cdot  \exp \left( \exp\left(   t^2 s \cdot 2^{O(s_0 \ell /\epsilon)}  \cdot \exp(\log^*n) +t\log s\right) \right),
\end{eqnarray*}
where $T'_{\text{pre}},T'$ denote the preprocessing and running times of the local list recovery algorithm for $C_n $, respectively.
\end{proof}

\subsection{Nearly-linear time decoding up to half the GV bound}\label{subsec:near-linear-GV}

Next we show how the nearly-linear time capacity achieving list recoverable codes of Theorem \ref{thm:near-linear-cap-list-rec} can be used to obtain codes of rate up to $0.02$ that are uniquely decodable up to half the Gilbert-Varshamov (GV) bound in nearly-linear time.  Let $H_2^{-1}: [0, 1] \to [0, \frac 1 2]$ be the inverse of the binary entropy function $H_2$ in the domain $[0,\frac{1}{2}]$.

\begin{theorem}[Nearly-linear time decoding up to half GV bound]\label{thm:near-linear-GV}
Choose constants $\rho \in [0,0.02]$ and $\epsilon >0$. Then there exists a randomized polynomial time algorithm which  for infinitely many $n$, given an input string $1^n$, outputs a description of a code $C_n$ that satisfies the following properties with probability at least $1-\exp(-n)$.\footnote{The randomized algorithm can output different codes under different random choices, we are only guaranteed that the output code has the required  properties with high probability. Also the encoding and decoding algorithms need access to the random choices made during the construction of the code.}

\begin{enumerate}
\item $C_n: \F_2^{\rho n} \to \F_2^n$ is a linear code of rate $\rho$ and relative
distance at least $\delta:=H_2^{-1}(1-\rho)-\epsilon$.
\item $C_n$ can be uniquely decoded probabilistically (with success probability $2/3$) from $\delta/2$ fraction of errors in time $n^{1+O(1/t)}=n^{1+o(1)}$ for
$t = O\inparen{ \frac{\log\log\log(n)}{ (\log\log\log\log(n))^2 }}$.
\item   $C_n$ can be encoded in time $n^{1+O(1/t)}=n^{1+o(1)}$ with the same choice of $t$. 
\end{enumerate}
 \end{theorem}

To prove the above theorem we rely on the following lemma which says that one can turn a code that approximately satisfies the Singleton bound into one that approximately satisfies the GV bound via random concatenation. The proof is similar to that of Thommesen~\cite{Thomm83}. In what follows let $\theta(x):=1-H_2(1-2^{x-1})$ for $x \in [0,1]$.

\begin{lem}\label{lem:random-concat} There is a constant $c$ so that the following holds. Choose $\epsilon>0$,
$ \rho_{\inn} \in [0, 1],$ and $ \rho_{\out} \in \left[0,\frac{\theta(\rho_{\inn})-\epsilon/2}{\rho_{\inn}}\right]$. Suppose that $s\ge\frac{c \cdot \rho_{\inn}}{\epsilon^{2}\cdot(1-\rho_{\out})}$.
Then the following holds for sufficiently large $n$. Let $C_{\out}: \F_{2^s}^{ \rho_{\out} \cdot n} \to \F_{2^s}^n$ be an
$\F_{2}$-linear code of rate $\rho_{\out}$ and relative distance at least $1-\rho_{\out}-\frac{\epsilon^{2}}{c}$.
Let $C: \F_{2}^{s \cdot \rho_{\out} \cdot n} \to \F_2^{s n/\rho_{\inn}}$ be a code obtained from $C_{\out}$  by applying
a random linear code $C^{(i)}: \F_{2}^s \to \F_{2}^{s/\rho_{\inn}}$
on each coordinate $i\in[n]$ of $C_{\out}$ independently (where we identify the field $\F_{2^s}$ with the vector space $\F_2^s$ via the usual $\F_2$-linear transformation).
Then $C$ has
relative distance at least $H_2^{-1}(1-\rho_{\out} \cdot \rho_{\inn})-\epsilon$ with probability
at least $1-exp(-n)$.
\end{lem}

We shall also use the following lemma that states the effect of concatenation on list recovery properties.

\begin{lem}\label{lem:concat-list-rec}  Let $C_{\out}: \F_{q^s}^{\rho_{\out} \cdot n} \to \F_{q^s}^n$ be an
 $(\alpha_{\out}, \ell_{\out}, L_{\out})$-list recoverable code, with a list recovery algorithm running in time $T_{\out}$.
Let  $C:  \F_{q}^{s \cdot \rho_{\out} \cdot n} \to \F_q^{s n/\rho_{\inn}}$ be the code obtained from $C_{\out}$
 by applying a code $C^{(i)}: \F_q^s \to \F_q^{s/\rho_{\inn}}$ on each coordinate $i \in [n]$ of $C_{\out}$. Suppose furthermore that
 at least $1-\epsilon$ fraction of the codes $C^{(i)}$ are $(\alpha_{\inn}, \ell_{\inn}, L_{\inn})$-list recoverable for $L_{\inn} = \ell_{\out}$, with a list recovery algorithm running in time $T_{\inn}$.
Then $C$ is $((\alpha_{\out}-\epsilon) \cdot \alpha_{\inn}, \ell_{\inn}, L_{\out})$-list recoverable in time $T_{\out} + O(n \cdot T_{\inn})$. \end{lem}

Before we prove the above pair of lemmas we show how they imply Theorem \ref{thm:near-linear-GV}.

\begin{proof}[Proof of Theorem \ref{thm:near-linear-GV}]
Let $c$ be a sufficiently large constant for which both Theorem \ref{thm:near-linear-cap-list-rec} and Lemma \ref{lem:random-concat} hold.  Let $\rho_{\inn}:=\theta^{-1}(\rho+\epsilon/2)$ and $\rho_{\out}:=\frac{\rho}{\rho_{\inn}}=\frac{\rho} {\theta^{-1}(\rho+\epsilon/2)}$.
Define
\[ t = t_n = O \inparen{ \frac{\log\log\log(n)}{(\log\log\log\log(n))^2}}\]
as in the statement of Theorem~\ref{thm:near-linear-GV},
and let $$s =  \max \big\{(ct/\epsilon^2)^{ct},  c^3  2^{1/\epsilon} t /\eps^4\big\} \cdot s_0/ (\rho_{\out} (1-\rho_{\out}))$$
for even
\begin{align*}
s_0 &:= \max \left\{\frac{c^2}{\eps^2} \inparen{1-\rho_{\out}-\frac{\epsilon^2}{2c}} \log\inparen{\frac{1}{1-\rho_{\out}-\epsilon^2/(2c)}} ,\right.\\
 &\qquad \qquad \qquad\qquad\qquad \left. \frac{c^2}{\epsilon^2} \left(\rho_{\out}+\frac{\epsilon^2}{2c}\right) \log\left(\frac{1}{\rho_{\out}+\epsilon^2/(2c)}\right) , \frac{c^2}{\epsilon^3} \right\}.
\end{align*}
Notice that chosing when $t = t_n$ is super-constant, as above, and $\epsilon, \rho$ are constant (as in the theorem statement) that we always have $s = t^{O(t)}$.

Choose any $n \in \mathbb{N}$ so that Theorem~\ref{thm:near-linear-cap-list-rec} guarantees that codes of block length $n$ exist (there are infinitely many of these).
Then the code $C_n$ is constructed as follows.

Let $C_{\out}: \F_{2^s}^{\rho_{\out} \cdot (\rho_{\inn} n/s)} \to \F_{2^s}^{\rho_{\inn} n/s}$ be the $\F_2$-linear code of rate
$\rho_{\out}$ and relative distance at least $1 - \rho_{\out} - \frac{\epsilon^2}{c}$ that is $\left(1 -\rho_{\out}-\frac{\epsilon^2} {c}, 2^{1/\epsilon}, n^{O(1/t)}\right)$-list recoverable in time $n^{1+O(1/t)}$
whose existence is guaranteed by Theorem \ref{thm:near-linear-cap-list-rec}.
Notice that here $t$ is growing with $n$, but slowly enough that we may apply Theorem~\ref{thm:near-linear-cap-list-rec}.

Let $C_n: \F_{2}^{\rho_{\out} \cdot \rho_{\inn} n} \to \F_2^n$ be a binary linear code of rate  $\rho_{\inn} \cdot \rho_{\out}=\rho$
obtained from $C_{\out}$ by applying a random linear code
$C^{(i)}:  \F_2^s \to \F_2^{s/\rho_{\inn}}$
on each coordinate $i \in [\rho_{\inn} n/s]$ of $C_{\out}$ independently.

Then by Lemma \ref{lem:random-concat} the code $C_n$ has relative distance at least $H_2^{-1}(1-\rho) -\epsilon$ with probability at least $1-\exp(-n)$.
Moreover, by Theorem \ref{thm:cap-list-dec} we also have that each $C^{(i)}$ is $(H_2^{-1}(1-\rho_{\inn} -\epsilon),2^{1/\epsilon})$-list decodable with probability at least $1-\exp(-s)= 1-o(1)$, so with probability at least $1-\exp(-n)$ it holds that at least $1-\epsilon^2/c$ fraction of the $C^{(i)}$'s are $(H_2^{-1}(1-\rho_{\inn} -\epsilon),2^{1/\epsilon})$-list decodable. Lemma \ref{lem:concat-list-rec} then implies that in this case the code $C_n$ is $((1-\rho_{\out} -2\epsilon^2/c )\cdot H_2^{-1}(1-\rho_{\inn} -\epsilon),n^{O(1/t)} )$-list decodable (probabilistically) in time $n^{1+O(1/t)}+n\cdot 2^{O(s)} =n^{1+O(1/t)}$.

Now observe that by Theorem \ref{thm:near-linear-cap-list-rec}
 the code $C_{\out}$ is encodable in time $n^{1+O(1/t)}$, and so encoding time of $C_n$ is $n^{1+O(1/t)}+n\cdot 2^{O(s^2)} =n^{1+O(1/t)} $. Consequently, whenever the list decoding radius of $C_n$ exceeds half the minimum distance, one can uniquely decode $C_n$ up to half the minimum distance in time $n^{1+O(1/t)}$
 by list decoding $C_n$, computing the codewords that correspond to the messages in the output list, and returning the (unique) closest codeword.

So we obtained that the code $C_n$ is uniquely decodable up to half the minimum distance in time $n^{1+O(1/t)}$ whenever
 $$  \left(1-\rho_{\out} -2\epsilon^2/c \right) \cdot H_2^{-1}(1-\rho_{\inn} -\epsilon)  \geq \frac{H_2^{-1}(1-\rho)-\epsilon} {2}.$$
Finally, it is shown in \cite{R07a}, Section 4.4 that the inequality above holds for sufficiently small constant $\epsilon >0$ by choice of $\rho=\rho_{\inn} \cdot \rho_{\out} \leq 0.02 $.
\end{proof}

It remains to prove Lemmas \ref{lem:random-concat} and \ref{lem:concat-list-rec}. We start with Lemma \ref{lem:random-concat}.

\begin{proof}[Proof of Lemma \ref{lem:random-concat}]
The proof follows the arguments of \cite{Thomm83}.
For a string $x$ of length $n$ let the {\em relative weight}
 $\wt(x)$ of $x$ denote the fraction of non-zero coordinates of $x$. It is well known that the relative distance of an $\F$-linear code equals the minimum relative weight $\wt(c)$ of a non-zero codeword $c \in C$.

Fix a codeword $c' \in C_{\out}$ with $\wt(c') =  \gamma \geq 1-\rho_{\out}- \epsilon^2/c$, and  let $c \in \F_2^{sn/\rho_{\inn}}$ be a word obtained
from $c'$ by applying a  random linear code $C^{(i)}:\F_{2}^s\to \F_{2}^{s/\rho_{\inn}}$ on each coordinate $i\in [n]$ of $c'$ independently.
Then for each  non-zero coordinate $i$ of $c'$ it holds that the $i$-th block of $c$ of length $s/\rho_{\inn}$ is distributed uniformly over $\F_2^{s/\rho_{\inn}}$, and so $\gamma s n/\rho_{\inn}$ coordinates of $c$ are uniformly distributed (while the rest equal zero).
Consequently, we have that
$$
 \Pr[\wt (c) < \delta ]\  \leq   {\gamma s n/\rho_{\inn} \choose \leq \delta s n/\rho_{\inn}} 2^{- \gamma s n/\rho_{\inn}}
 \leq  2^{ H_2\left ({\delta } /{\gamma}\right)\gamma sn/\rho_{\inn}}\cdot2^{-\gamma sn/\rho_{\inn}}.
$$

Next we apply a union bound over all codewords $c' \in C_{\out}$. For this fix $\gamma> 0$ such that $\gamma \geq 1-\rho_{\out}- \epsilon^2/c$ and $\gamma n \in \mathbb{N}$. Then it holds that the number of codewords in $C_{\out}$ of relative weight $\gamma$ is at most
\[
{n \choose \gamma n}\cdot \left(2^s\right)^{\gamma n - (1-\rho_{\out}- \epsilon^2/c)n}\leq 2^{n}\cdot 2^{(\gamma - (1-\rho_{\out}- \epsilon^2/c))sn},
\]
where the above bound follows since there are at most ${n \choose \gamma n}$ choices for the location of the non-zero coordinates, and for any such choice fixing the value of  the first
$\gamma n - (1-\rho_{\out}- \epsilon^2/c)n$ non-zero coordinates determines the value of the rest of the non-zero coordinates (since two different codewords cannot differ on less than $(1-\rho_{\out}- \epsilon^2/c)n$ coordinates).

Consequently, we have that
$$
\Pr[\dist(C)  <  \delta ]  \leq  \sum_{1-\rho_{\out}-\epsilon^2/c \le \gamma \le 1,\;\gamma n\in\mathbb{N}}  2^{n}\cdot 2^{(\gamma - (1-\rho_{\out}- \epsilon^2/c))sn} \cdot  2^{ H_2\left ({\delta } /{\gamma}\right)\gamma sn/\rho_{\inn}}\cdot2^{-\gamma sn/\rho_{\inn}}
$$
$$ =    \sum_{1-\rho_{\out}-\epsilon^2/c\le \gamma \le 1,\;\gamma n\in\mathbb{N}}
\exp\left[- \gamma s n/\rho_{\inn}
\left(1- \rho_{\inn}\cdot\left(1-\frac{1-\rho_{\out}-\epsilon^2/c} {\gamma}\right)
-  \frac {\rho_{\inn}} { \gamma s} -H_2\left( \frac { \delta} {\gamma}\right)
  \right) \right].
  $$
Finally, Lemma 8 in \cite{GR08} implies that in our setting of parameters, and for choice of $\delta := H_2^{-1}(1-\rho_{\inn} \cdot \rho_{\out})-\epsilon$, the right hand side of the above inequality is at most $\exp(-n)$, which completes the proof of the lemma. The running time follows.
\end{proof}

Next we prove Lemma \ref{lem:concat-list-rec}.

\begin{proof}[Proof of Lemma \ref{lem:concat-list-rec}]
For $i \in [n]$ and $j \in [s/\rho_{\inn}]$, let $S_{i,j} \subseteq \F_q$ be a list of at most $\ell_{\inn}$ possible symbols for the coordinate
$ C(x)_{i,j} :=   C(x)_{(i-1)\cdot( s/\rho_{\inn}) +j},$ which is the $j$-th coordinate of $C^{(i)} (C_{\out}(x)_i )$.

Suppose that for at most a $(\alpha_{\out} -\epsilon) \cdot \alpha_{\inn}$ fraction of coordinates $(i,j)$, $ C(x)_{i,j} \notin S_{i,j}$.  Then by Markov's inequality, for at most a $\alpha_{\out} -\epsilon$ fraction of $i \in [n]$, the blocks $C^{(i)} (C_{\out}(x)_i )$ have more than $\alpha_{\inn}$ fraction of the $j \in [s/\rho_{\inn}]$ so that $C(x)_{i,j} \notin S_{i,j}$.
Thus, we may list recover each block $C^{(i)} (C_{\out}(x)_i )$ which is list recoverable
to obtain a list $S_i \subseteq \F_{q^s}$ of at most $L_{\inn}=\ell_{\out}$ possible symbols for $C_{\out}(x)_i$,  and the above reasoning shows that $C_{\out}(x)_i \notin S_i$ for at most $\alpha_{\out} n$ values of $i$.  Now we may run $C_{\out}$'s list recovery algorithm to obtain a final list of size $L_{\out}$.
\end{proof}

\paragraph*{Acknowledgements.} The second author would like to thank Swastik Kopparty for raising the question of obtaining capacity achieving locally list decodable codes which was a direct trigger for this work as well as previous work \cite{KMRS15-STOC,GKORS17}, and Sivakanth Gopi, Swastik Kopparty, Rafael  Oliveira and Shubhangi Saraf for many discussions on this topics. The current collaboration began during the Algorithmic Coding Theory Workshop at ICERM, we thank ICERM for their hospitality.

\bibliographystyle{alpha}
\bibliography{lr}

\appendix
\section{List-recovery of algebraic geometry codes}\label{app:ag}
In this appendix, we outline how the approach of~\cite{GX13} needs to be changed in order to obtain linear list-recoverable codes.  The main theorem is as follows.

\begin{theorem}\label{thm:ag}
There are constants $c,c_0$ so that the following holds.
Choose $\eps > 0$ and a positive integer $\ell$.  Suppose that $q \geq \ell^{c/\eps}$ is an even power of a prime.
Let $N_0 = q^{c_0 \ell / \epsilon}$.

Then for all $N \geq N_0$, there is a deterministic polynomial-time construction of an $\F_q$-linear code $\cC: \F_q^{(1 - \eps)N} \to  \F_q^N$ of rate $1 - \eps$ and relative distance $\Omega(\eps^2)$ which is
$(\Omega(\eps^2), \ell, L)$-list-recoverable in time $\poly(N,L)$, returning a list that is contained in a subspace over $\F_q$ of dimension at most
\[ \inparen{ \frac{ q^{c\ell/\eps}}{\eps} }^{2^{\log^*(N)}} . \]
In particular, when $\eps, \ell,q$ are constant, the ouput list size $L$ is $\exp(\exp(\exp(\log^*N)))$.
\end{theorem}
 We remark that the list size is very slowly growing (although admittedly with extremely large constants).

We follow the approach of~\cite{GX13,GK14}.   In \cite{GX13}, Guruswami and Xing show how to construct high-rate list-decodable codes over a constant alphabet, modulo a construction of \em explicit subspace designs. \em  In \cite{GK14}, Guruswami and Kopparty gave such constructions and used them to construct high-rate list-decodable codes over constant-sized alphabets with small list-sizes.  We would like to use these codes here.  However, there are two things which must be modified.  First, the guarantees of~\cite{GX13,GK14} are for list-decodability, and we are after list-recoverability.  Fortunately, this follows from a standard modification of the techniques that they use.  Second, the codes that they obtain are not linear, but rather are linear over a subfield of the alphabet.  To correct this, we concatenate these codes with list-recoverable linear codes of a constant length.  A random linear code has this property, and since we only require them to be of constant length, we may find such a code, and run list-recovery algorithms on it, in constant time.

We begin by addressing the leap from list-decodability to list-recovery, and then discuss the code concatenation step.
We refer the reader to~\cite{GX13,GK14} for the details (and, indeed, for several definitions); here we just outline the parts which are important for list-recovery.
The basic outline of the construction (and the argument) is as follows:
\begin{description}
	\item[Step 1.] Show that AG codes are list-decodable, with large but very structured lists.  We will extend this to list-recoverability with structured lists.
	\item[Step 2.] Show that one can efficiently find a subcode of the AG code which will avoid this sort of structure: this reduces the list size.  This part of the argument goes through unchanged, and will yield a list-recoverable code over $\F_{q^m}$ with small list size.
\end{description}
Once we have $\F_q$-linear codes over $\F_q^m$ that are list-recoverable, we discuss the third step:
\begin{description}
	\item[Step 3.]  The code produced is $\F_q$-linear (rather than $\F_{q^m}$-linear).  This was fine for~\cite{GX13,GK14}, but for us we require a code which is linear over the alphabet it is defined over.  To get around this we concatenate the codes above with a random linear code of length $m$ over $\F_q$.  This will result in an $\F_q$-linear code over $\F_q$ that is list-recoverable with small list sizes.
\end{description}

We briefly go through the details.  First we give a short refresher/introduction to the notation.  Then we handle the three steps above, in order.  We note that throughout this appendix we will refer to Theorem and Lemma numbers in the extended version~\cite{GX13eccc} rather than the conference version~\cite{GX13}.

\paragraph{Step 0. Algebraic Geometry Codes and basic notation.}
Since we do not need to open up the AG code machinery very much in order to extend the results of \cite{GX13} to list-recovery, we do not go into great detail here, and we refer the reader to~\cite{GX13} and the references therein for the technical details, and to~\cite{stichtenoth} for a comprehensive treatment of AG codes.  However, for the ease of exposition here (for the reader unfamiliar with AG codes), we will introduce some notation and explain the intuitive definitions of these notions.  In particular, we will use the running example of a rational function field.  We stress that this is \em not \em the final function field used; thus the intuition should be taken as intuition only.

Let $F/\F_q$ be a function field of genus $g$.
One example (which may be helpful to keep in mind) of a genus-$0$ function field is the
rational function field $\F_q(X)/\F_q$, which may be thought of as rational functions $f(X)/g(X)$, where $f,g \in \F_q[X]$ are irreducible polynomials.
For the code construction, we will use a function field of larger genus (given by the Garcia-Stichtenoth tower, as in~\cite{GX13}), but we will use this example to intuitively define the algebraic objects that we need.

Let $P_\infty, P_1,\ldots,P_n$ be $n+1$ distinct $\F_q$-rational places (that is, of degree $1$).
Formally, these are ideals, but they are in one-to-one correspondence with $\F_{q} \cup \inset{ \infty }$, and let us think of them that way.  For each such place $P$, there is a map (the \em residue class map \em with respect to $P$) which maps $F/\F_q$ to $\F_q$; we may think of this as function evaluation, and in our example of $\F_q(X)/\F_q$, if $P$ is a place associated with a point $\alpha \in \F_q$, then indeed this maps $f(X)/g(X)$ to $f(\alpha)/g(\alpha)$.

Let $\mathcal{L}(lP_\infty)$ be the Riemann-Roch space over $\F_q$.  Formally, this is
\[ \mathcal{L}(lP_\infty) = \inset{ h \in F\setminus \inset{0} : \nu_{P_\infty}(h) \geq -l } \cup \inset{0}, \]
where $\nu_{P_\infty}$ is the discrete valuation of $P_\infty$.  Informally (in our running example), this should be thought of as the set of rational functions $f(X)/g(X)$ so that $\deg(g(X)) - \deg(f(X)) \geq -l$.  In particular, the number of poles of $f/g$ is at least the number of roots, minus $l$.    It would be tempting, in this example, to think of these as degree $\leq l$ polynomials; all but at most $l$ of the powers of $X$ in the numerator are ``canceled" in the denominator.  Of course, there are many problems with this intuition, but it turns out that
this indeed works out in some sense.  In particular, it can be shown that the dimension of this space is at least $l - g + 1$.  When $g = 0$ (as in our running example), it is exactly $l + 1$, the same as the dimension of the space of degree-$\leq l$ polynomials.

More generally (whatever the genus), for any rational place $P$, we may write a function $h \in \mathcal{L}_m(lP_\infty)$ as
\begin{equation}\label{eq:expand}
 h = \sum_{j=0}^\infty h_j T^j,
\end{equation}
where $T$ is a local parameter of $P$, and it turns out that $h$ is uniquely determined by the first $l + 1$ coefficients $h_0, h_1,\ldots, h_{l + 1}$.

Now let $F_m$ be the constant extension $\F_{q^m}\cdot F$, and let $\mathcal{L}_m(lP_\infty)$ be the corresponding Riemann-Roch space.  This has the same dimension over $\F_{q^m}$ as $\mathcal{L}(lP_\infty)$ does over $\F_q$.  Now we consider the algebraic geometry code defined by
\[ \cC(m;l) := \inset{ (h(P_1),\ldots,h(P_n)) \suchthat h \in \mathcal{L}_m(lP_\infty) }. \]
Following the intuition that $h(P_i)$ denotes function evaluation, this definition looks syntactically the same as a standard polynomial evaluation code, and should be thought of that way.  This is an $\F_{q^m}$-linear code over $\F_{q^m}$, with block length $n$ and dimension at least $l - g + 1$.

\paragraph{Step 1. List-decoding with structured lists to list-recovery with structured lists.}
With the preliminaries (and some basic, if possibly misleading, intuition for the reader unfamiliar with AG codes) out of the way, we press on with the argument.

Fix a parameter $k$, and consider
a general AG code $\cC(m;k + 2g - 1)$, with the notation above.  (We will fix a particular AG (sub)code later, by choosing a function field and by choosing a subcode).
Let $S_1,\ldots,S_n \subset \F_{q^m}$ be lists of size at most $\ell$ corresponding to each coordinate.  We first show that $\cC(m;k + 2g - 1)$ is $(1 - \alpha, \ell, L)$-list-recoverable for some $\alpha$ to be chosen below, where the list size is very large, but the list is structured.
	In~\cite{GX13}, the approach (similar to that in~\cite{G11} or \cite{GW11}) is as follows.
\begin{enumerate}
	\item We will first find a low-degree interpolating linear polynomial (whose coefficients live in Riemann-Roch spaces)
	\[ Q(Y_1,\ldots, Y_s) = A_0 + A_1 Y + \cdots + A_s Y_s \]
	so that $A_i \in \mathcal{L}_m(DP_\infty)$ and $A_0 \in \mathcal{L}_m((D + k + 2g - 1)P_\infty)$, for some parameter $k$ to be chosen later, for
\[ D = \lfloor \frac{ \ell n - k + (s - 1)g + 1 }{s + 1} \rfloor,\]
and subject to $\ell n$ linear constraints over $\F_{q^m}$.  Before we list the constraints, notice that the number of degrees of freedom in $Q$ is
\[ s(D - g + 1) + D + k + g, \]
because the $\F_{q^m}$-dimension of $\mathcal{L}_m((D + k + 2g - 1)P_\infty)$ is at least $D + k + g$, and the $\F_{q^m}$-dimension of $\mathcal{L}_m((DP_\infty))$ is at least $D - g + 1$.  Thus, the choice of $D$ shows that the dimension of this space of interpolating polynomials is greater than $\ell n$.  Thus, we will be able to find such a $Q$ that satisfies the $\ell n$ following $\ell n$ constraints.
For each $i \in [n]$ and for all $y \in S_i$, we have the constraint that
\[ A_0(P_i) + A_1(P_i)y + A_2(P_i)y^q + \cdots + A_s(P_i)y^{q^{s-1}} = 0. \]
	\item With this polynomial $Q$ in hand, we observe that if $h \in \mathcal{L}_m((k + 2g - 1)P_\infty)$ whose encoding has $h(P_i) \in S_i$ for at least $\alpha n$ positions $i$, for $\alpha n > D + k + 2g - 1$, then $Q(h, h^\sigma, \ldots, h^{\sigma^{s-1}}) = 0$,
where $h^\sigma$ denotes the extension of the Frobenius automorphism $\alpha \mapsto \alpha^q$ on $\F_{q^m}$ to $\mathcal{L}_m(lP_\infty)$.
This proof (Lemma 4.7 in \cite{GX13eccc}) remains unchanged when we pass to list-recovery from list-decoding.  Briefly, this agreement means that
\[ Q(h, \ldots, h^{\sigma^{s-1}})(P_i) = A_0(P_i) + A_1(P_i) h(P_i) + \cdots + A_s(P_i) h(P_i)^{q^{s-1}} = 0 \]
for at least $\alpha n$ values of $i$, and so the function $Q(h, h^\sigma, \ldots, h^{\sigma^{s-1}})$ (which lies in $\mathcal{L}_m((D + k + 2g - 1)P_\infty)$; as per the intuition above, we are thinking of these as roughly analogous to degree-$(D + k + 2g -1)$ polynomials) has at least $\alpha n \geq D + k + 2g - 1$ roots, and hence is the zero function.
	\item Thus, any element $h \in \mathcal{L}_m((k + 2g - 1)P_\infty)$ that agrees with at least $\alpha n$ lists also satisfies $Q(h, \ldots, h^{\sigma^{s-1}}) = 0$.  It remains to analyze the space of these solutions, and to show that they are nicely structured.  This requires one more step, which goes through without change.
	More precisely, \cite{GX13} takes a subcode of $\cC(m; k + 2g - 1)$; this subcode will still have a large list size, but the list will be structured.
This resulting code, denoted $\cC(m; k + 2g - 1 | \F_{q^m}^k)$, has dimension $k$.  (Recall that $\cC(m; k + 2g - 1)$ has dimension $k + g$, so we have reduced the dimension by $g$.)
We refer the reader to \cite{GX13} for the details, as they do not matter for us.
At the end of the day,
the analysis of~\cite{GX13} (Lemma 4.8 in the full version~\cite{GX13eccc}) applies unchanged to show that the set of messages $h$ in this new code that are solutions to this equation lie in a structured space: more precisely, the coefficients $(h_0, h_1, \ldots, h_{k + 2g - 1})$ as in \eqref{eq:expand} belong to an \em $(s-1, m)$-ultra-periodic subspace \em of $\F_q^{m(k + 2g - 1)}$.
For us, the precise definition of this does not matter, as we may use the rest of \cite{GX13} as a black box.
	\item Before we move on, we summarize parameters.  We have so far established that there is a code $\cC(m;k + 2g - 1 | \F_{q^m}^k)$ that is list-recoverable with agreement parameter $\alpha$ and inner list sizes $\ell$, resulting in a structured list.  The requirement on $\alpha$ is:
\begin{align*}
	\alpha n & > D + k + 2g - 1 \\
		& = \left\lfloor \frac{ \ell n - k + (s-1)g + 1 }{s + 1} \right\rfloor + k + 2g - 1,
\end{align*}
and so it suffices to take
\begin{align*}
	\alpha n
		& >  \frac{ \ell n - k + (s-1)g + 1 }{s + 1} + k + 2g - 1 \\
	&= \frac{1}{s + 1} \inparen{ \ell n + s(k-1) + g(3s + 1) }.
\end{align*}
Again, the dimension of the code is $k$ and the length is $n$.  It is $\F_{q^m}$-linear over $\F_{q^m}$.
\end{enumerate}
\paragraph{Step 2. Taking a subcode.}
For this step, we may follow the argument of \cite{GX13} without change.
Briefly, to instantiate the AG code we use a function field from a \em Garcia-Stichtenoth tower. \em  The parameters of this are as follows: we choose a prime power $r$, and let $q = r^2$.  Then we choose an integer $e >0$.  There is a function field $F = K_e$ so that
$K_e$ has at least $n = r^{e-1}(r^2 - r) + 1$ rational places, and genus $g_e$ bounded by $r^e$.   This is the function field we will use.
We remark that \cite{GX13} has to do a bit of work here to show that one can actually find a description of the structured list efficiently, but it can be done.  We plug in parameters to obtain the following Lemma, which is analogous to Theorem 4.14 in \cite{GX13eccc}.

\begin{lem}\label{lem:almostthere}
Let $q$ be the even power of a prime, and choose $\ell, \eps > 0$.   There is a parameter $s = O(\ell /\eps)$ so that the following holds.  Let $m \geq s$ and let $R \in (0,1)$.
Suppose that $\alpha \geq R + \eps + 3/\sqrt{q}$.
Then for infinitely many $n$ (all integers of the form $n = q^{e/2}(\sqrt{q} - 1)$), there is a deterministic polynomial-time construction of an $\F_{q^m}$-linear code $\cC$ of block length $n$, dimension $k = Rn$, so that the following holds: for any sets $S_1,\ldots, S_n \subseteq \F_{q^m}$ with $|S_i| \leq \ell$ for all $i$, the set of messages leading to codewords $c \in \cC$ so that $c_i \in S_i$ for at least $\alpha n$ coordinates $i$ is contained in one of $q^{O(mn)}$ possible $(s - 1, m)$-ultra periodic $\F_q$-affine subspaces of $\F_q^{mk}$.  Further, this collection of subspaces can be described in time $\poly(n, m)$.
\end{lem}

\begin{proof}
Our condition on $\alpha$ is that it is at least
\begin{align*}
\frac{ \ell n + s(k-1) + g_e(3s + 1) }{n (s + 1) }
&\leq \frac{ \ell n + s(k-1) + n(3s + 1)/(r - 1) }{n(s+1)} \qquad \text{Using $g_e \leq n/(r-1)$}\\
&= \frac{ \ell + s(R - 1/n) + (3s + 1)/(r - 1) }{s + 1 }.
\end{align*}
Choosing $s = O(\ell/\eps)$ and using the fact that $r = \sqrt{q}$ gives the conclusion.
\end{proof}

With this lemma in hand, we may proceed exactly as the proof in \cite{GX13}; indeed, it is exactly the same code, and we exactly the same conclusion on the structure of the candidate messages.  The basic idea is to choose a subset of messages carefully via a \em cascaded subspace design. \em  This ensures that the number of legitimate messages remaining in the list is small, and further that they can be found efficiently.

We briefly go through parameters, again referring the reader to the discussion in \cite{GX13,GK14} for details.
We will fix
\begin{equation}\label{eq:chooses}
s = O(\ell/\eps), \qquad \text{and} \qquad  m = O\inparen{ \frac{ \ell}{\eps^2} \cdot \log_q(\ell/\eps) }.
\end{equation}
We now trace these choices through the analysis of \cite{GX13, GX14}.
\begin{rem} The reader familiar with these sorts of arguments might expect us to set $m = \ell/\eps^2$, and indeed this would be sufficient if we could allow $q$ to be sufficiently large.  However, in this case, setting $m$ this way would result in a requirement that $q \geq \ell /\eps^2$.  We would like $q$ to be independent of $\ell$ for the next concatenation step to work (of course, the alphabet size $q^m$ must be larger than $\ell$), and this requires us to take $m$ slightly larger.  This loss comes out in the final list size.
\end{rem}

Without defining a cascaded subspace design, we will just mention that it is a sequence of $T$ subspace designs (which we will not define either); a cascaded subspace design comes with vectors of parameters $(r_0, \ldots, r_T)$, $(m_0,\ldots, m_T)$, and $(d_0,\ldots,d_{T_1})$.  For $i=1,\ldots,T$, the $i$'th subspace design in this sequence is a $(r_{i-1},r_i)$-strong-subspace design in $\F_q^{m_{i-1}}$, of cardinality $m_i/m_{i-1}$, and dimension $d_{i-1}$.
Again, for us it does not matter what a strong subspace design is, only that we may find explicit ones:
\begin{theorem}[Follows from Theorem 6 in \cite{GK14}]\label{thm:gk}
For all $\zeta \in (0,1)$ and for all $r,m$ with $r \leq \zeta m / 4$, and for all prime powers $q$ so that $2r/\zeta < q^{\zeta m/(2r)}$, there exists an explicit collection of $M \geq q^{\Omega( \zeta m  / r ) }/(2r)$ subspaces in $\F_q^m$, each of codimension at most $\zeta m$, which form a $(r, r^2/\zeta)$-strong subspace design.
\end{theorem}
\begin{rem} In \cite{GK14}, the theorem is stated for $(r,r/\zeta)$-weak subspace designs; however, as is noted in that work, a $(A,B)$-weak subspace design is also a $(A, AB)$-strong subspace design, which yields our version of the theorem.
\end{rem}

Below, we will use Theorem~\ref{thm:gk} in order to instantiate a cascaded subspace design.  The reason we want to do this is because of Lemma 5.6 in \cite{GX13eccc}:
\begin{lemma}[Lemma 5.6 in \cite{GX13eccc}]\label{lem:canonicalsubspace}
Let $\mathcal{M}$ be a $(r_0, r_1, \ldots , r_T)$-cascaded subspace design with length-vector $(m_0,m_1,\ldots,m_T)$.
Let $A$ be an $(r,m)$-ultra periodic affine subspace of $\mathbb{F}_q^{m_T}$.  Then the dimension of the affine space $A \cap U(\mathcal{M})$ is at most $r_T$, where $U(\mathcal{M})$ denotes the \em canonical subspace \em of $\mathcal{M}$.
\end{lemma}

We have not defined a canonical subspace, and we refer the reader to \cite{GX13eccc} for details; the important thing for us is that
we wish to construct a cascaded subspace design $\mathcal{M}$ so that $r_T$ is small, $m_T$ is equal to $mk$, and so that $r_0 = s-1$ and $m_0 = m$.  This will allow us to choose a subcode of the code from Lemma~\ref{lem:almostthere} by restricting the space of messages to the canonical subspace $U(\mathcal{M})$, and this will be the $\F_q$-linear code (over $\F_q^m$) that we are after.

We may use Theorem~\ref{thm:gk} to instantiate such a cascaded subspace design as follows (the derivation below follows the proof of Lemma 5.7 in \cite{GX13eccc}).
We choose $\zeta_i = \eps / 2^i$, $r_0 = s-1$, and $r_i = r_{i-1}^2 / \zeta_i$.  We choose
$m_0 = m$ and we will define $m_i = m_{i-1} \cdot q^{\sqrt{m_{i-1}}}$.
We will continue up to $i = T$, choosing $T$ so that $m_T = mk$.
At this point, we must deal with the detail that there may be no such $T$; to deal with this we do exactly as in the proof of Lemma 5.7 in~\cite{GX13eccc}  and modify our last two choices of $m_{T-1}, m_T$ so that $m_T \leq mk$ but is close (within an additive $\log^2_q(km)$); for our final subspace, we will pad the $m_T$-dimensional vectors with $0$'s in order to form a subspace in $\F_q^{mk}$ with the same dimension.
  Choosing $m_T \approx mk$ puts $T = O(\log^*(mk))$, and $r_T = O(s/\eps)^{2^T} \leq (\ell/\eps)^{O(\log^*(mk))}.$

With these choices, we instantiate $T$ subspace designs via Theorem~\ref{thm:gk}, with $m \gets m_i, r \gets r_i$, and $\zeta \gets \zeta_i$.  We check that the requirements of Theorem~\ref{thm:gk} are satisfied, beginning with the requirement that $r_i \leq \zeta_i m_i / 4.$
Since $m_i \zeta_i$ grows much faster than $r_i$ as $i$ increases, it suffices to check this for $i = 0$, when we require $r_0 \leq \zeta_0 m_0$, or $s-1 \leq m \eps / 8$.  Our choices of $m$ and $s$ in \eqref{eq:chooses} satisfy this.

The next requirement is that $2r_i/\zeta_i \leq q^{\zeta_i m_i / (2r_i)}$ for all $i$.  Again, the right hand side grows much faster than the left, and so we establish this for $i=0$, requiring that
\[ \frac{ 4 (s-1) }{\eps} \leq q^{\eps m / 4(s-1)}. \]
With our choices of $m$ and $s$, this requirement is that
\[ \frac{\ell}{\eps^2} \leq q^{O( \log_q(\ell/\eps) ) } , \]
which is true.

Thus, Theorem~\ref{thm:gk} provides us with a cascaded subspace design with the given parameters.  As mentioned above, we may then use Lemma~\ref{lem:canonicalsubspace} to choose an appropriate subcode of our AG code from Lemma~\ref{lem:almostthere}.
We have chosen the parameters above so that $(r_0,m_0) = (s-1,m)$, precisely the guarantee of Lemma \ref{lem:almostthere}.  Thus, the final bound on the dimension of this intersection is $r_T \leq (\ell/\eps)^{O(\log^*(mk))}$, which gives our final list size.  Finally, we observe (as in Observation 5.5 of~\cite{GX13eccc}) that the dimension of the resulting subcode is at least $(1 - \sum_i \zeta_i) m_T = (1 - \eps)mk$.  Thus the final code has dimension at least $(1 - \eps)mk$ over $\F_q^{km}$, and hence the final rate is at least $(1 - \eps)R$.  Observing that $q$ must be at least $\eps^{-2}$ for the $1/\sqrt{q}$ term in Lemma~\ref{lem:almostthere} to be absorbed into the additive $\eps$ factor, we arrive at the following theorem.

\begin{theorem}\label{thm:fqlinear} Let $q$ be an even power of a prime, and choose $\ell, \eps > 0$, so that $q \geq \eps^{-2}$.  Choose $\rho \in (0,1)$.
There is an $m_{min} = O(\ell \log_q(\ell/\eps) /\eps^2)$ so that the following holds for all $m \geq m_{min}.$
For infinitely many $n$ (all $n$ of the form $q^{e/2} (\sqrt{q} - 1)$ for any integer $e$),
there is a deterministic polynomial-time construction of an $\F_q$-linear code $\cC:\F_{q^m}^{\rho n} \to  \F_{q^m}^n$ of rate $\rho$ and relative distance $1 - \rho - O(\eps)$
that is $(1 - \rho - \eps, \ell, L)$-list-recoverable in time $\poly(n, L)$, returning a list that is contained in a subspace over $\F_q$ of dimension at most
\[ \inparen{\frac{\ell}{\eps}}^{2^{\log^*(mk)}}. \]
\end{theorem}

We note that the distance of the code comes from the fact that it is a subcode of $C(m; k + 3g_e - 1)$, which has distance at least $n - (k + 2g - 1) = n - 2g - k + 1.$  In the above parameter regime, the genus $g_e$ satisfies $g_e \leq n/(r-1) = n/(\sqrt{q} -1 ) = O(\eps n)$.
Thus, the relative distance of the final code is at least $(n - 2g_e - k + 1)/n \geq 1 - O(\eps) - \rho.$

\paragraph{Step 3. Concatenating to obtain $\F_q$-linear codes over $\F_q$.}
Theorem \ref{thm:fqlinear} gives codes over $\F_{q^m}$ which are $\F_q$-linear.  For our purposes, to prove Theorem~\ref{thm:ag}, we require codes over $\F_q$ which are $\F_q$-linear.
Thus, we will concatenate these codes with random $\F_q$-linear codes from Corollary~\ref{cor:list-rec-random} and apply Lemma~\ref{lem:concat-list-rec} about the concatenation of list-recoverable codes.  In more detail, we choose parameters as follows.

Let $\eps > 0$ and let $\epsilon' = \epsilon/2$, and choose any integer $\ell$ and any block length $N$.  Fix a constant $c$ and parameters $m$ and $e$ which will be determined below.
Choose an even prime power $q$ so that
\[ q \geq \max \inset{ \ell^{c/\epsilon}, \epsilon^{-c} }.\]
Let $\cC_{in}$ be a random $q$-ary linear code of rate $\rho_{in} = 1 - \eps'$ of length $m/\rho_{in}$.
By Corollary~\ref{cor:list-rec-random}, there exists an $\F_q$ linear code $\cC_{in}$ with rate $\rho_{in} = 1 - \epsilon'$ and block length $m/\rho_{in}$ which is $(\alpha_{in}, \ell_{in}, L_{in})$-list-recoverable, for $\alpha_{in} = \epsilon'/2$, $\ell_{in} = \ell$, and $L_{in} = q^{2c\ell/\epsilon'}$.
We note that we can choose $c$ large
 enough to ensure that the hypothesis of Corollary~\ref{cor:list-rec-random} hold.

Let $\cC_{out}$ be the codes from Theorem~\ref{thm:fqlinear}, instantiated with rate $\rho_{in} = 1 - \epsilon'$, $\eps \gets \eps'/2$ and $\ell \gets L_{in}$.
With these parameters, we will get a code over $\F_{q^m}$ of length $n = q^{e/2}(\sqrt{q} - 1)$ which is $(\alpha_{out}, L_{in}, L_{out})$-list-recoverable,
where
\begin{align*}
L_{out} &= \exp_q\inparen{ (L_{in}/\epsilon')^{2^{\log^*(mk)}}} \\
&= \exp_q \inparen{ \inparen{ \frac{ q^{2c\ell/\eps'}}{\eps'} }^{2^{\log^*(mk)  }}}
\end{align*}
and where
\[ \alpha_{out}= 1 - \rho_{in} - \eps' =  \eps'/2. \]

Let $m_{\min}$ be as in Theorem~\ref{thm:fqlinear}, so that
\[ m_{\min} = O( L_{in} \log_q( L_{in}/\epsilon') / (\epsilon')^2)  = O\inparen{ \frac{q^{c\ell/\epsilon'} c \ell} {(\epsilon')^3}}.\]
We will choose $m$ so that
\begin{equation}
\label{eq:choosem}
m_{\min} \leq m \leq q \cdot m_{\min}.
\end{equation}
  Notice that, given the definition of $m_{\min} = O( q^{c\ell/\epsilon'} c\ell / (\epsilon')^3 )$, choosing $m$ slightly larger than $m_{\min}$---as large as $q \cdot m_{\min}$---amounts to replacing the constant $c$ with $c+1$.
  Thus, the choices of $m$ and $c$ (subject to \eqref{eq:choosem}) will not affect the list-recoverability of $\cC_{out}$, but they will affect the block length of the concatenated code.

  Formally, Lemma~\ref{lem:concat-list-rec} implies that the concatenated code has rate $\rho_{in} \cdot \rho_{out} = (1 - \eps')^2 \geq 1 - \eps$, and is $(\alpha_{in}\alpha_{out}, \ell, L_{out} )$-list-recoverable.  Here, we have
  \[ \alpha_{in}\alpha_{out}= (\eps')^2/4 = \Omega(\eps^2), \]
  which is what is claimed in Theorem~\ref{thm:ag}.
  The output list size claimed in Theorem~\ref{thm:ag} follows from the choice of $m$ and our guarantee on $L_{out}$.
  We note that the concatenated code will have message length $K = mk$, and so we write $\log^*(mk) = \log^*(K)$.

  Finally, we choose $m$ and $e$.  At this point, the choice of these parameters (subject to \eqref{eq:choosem}) will not affect that list-recoverability of the concatenated code, but they do control the block length of the code and the running time of the decoding algorithm.  The block length  is
  \[ \frac{m}{\rho_{in}} \cdot q^{e/2} (\sqrt{q} - 1). \]
    In order to prove that we can come up with such codes for all sufficiently large block lengths $N$, as required in the statement of Theorem~\ref{thm:ag}, we must show that for all sufficiently large $N$, we can choose $m$ satisfying \eqref{eq:choosem} and $e$ so that
\[ N = \frac{m}{\rho_{in}} \cdot q^{e/2} (\sqrt{q} - 1 ). \]
That is, we want to find an integer $e$ so that
\[ \frac{N \cdot (1 - \epsilon/2) }{q^{e/2}(\sqrt{q} - 1)} \in [ m_{\min}, q\cdot m_{\min}]. \]
However, we have chosen this window for $m$ to be large enough so that such an $e$ exists as long as $N$ is sufficiently large (in terms of $q, \ell, \epsilon$).  More precisely, for some large enough constant $C$, we require
\[ N \geq q^{C \ell / \epsilon },\]
which is our choice of $N_0$ in Theorem~\ref{thm:ag}.

Now we verify the running time of the list-recovery algorithm.  The outer code $C_{out}$ can be list-recovered in time $\poly(n,L_{out})$ by Theorem~\ref{thm:fqlinear}.  The inner code can be list-recovered by brute force in time $q^{O(m)} = \exp_q\inparen{O\inparen{q^{2(c+1)\ell / \eps} \ell/\eps^3 }} = \poly(L_{out})$.  Lemma~\ref{lem:concat-list-rec} implies that the final running time is $\poly(N, L)$, where $L = L_{out}$ is the final list size and $N$ is the block length of the concatenated code.


\end{document}